\PassOptionsToPackage{pdfpagelabels=true}{hyperref} 
\documentclass{sig-alternate-05-2015}

\usepackage[usenames]{color}
\usepackage{bm}
\usepackage{xspace}
\usepackage{paralist}
\usepackage{graphicx}
\usepackage{epstopdf}
\usepackage{boxedminipage}
\usepackage{tabularx}
\usepackage{array}
\usepackage{booktabs}
\usepackage{textcase}
\usepackage{multirow}
\usepackage{braket}
\usepackage{tikz}

\usepackage{rotating}
\usepackage{algorithm}
\usepackage{algpseudocode}
\usepackage{caption}
\usepackage{enumitem}
\usepackage{amsmath,amsfonts}

\usepackage{braket}

\usepackage{xcolor}

\usepackage{algorithm}
\usepackage{algpseudocode}

\usepackage[amsmath,thmmarks]{ntheorem}

\theoremstyle{plain}
\theoremheaderfont{\normalfont\sc}
\theorembodyfont{\normalfont\itshape}
\theoremseparator{.}
\theoremsymbol{}
\newtheorem{theorem}{Theorem}
\newtheorem{lemma}[theorem]{Lemma}

\theoremstyle{plain}
\theoremheaderfont{\normalfont\sc}
\theorembodyfont{\normalfont}
\theoremseparator{.}
\theoremsymbol{}

\newtheorem{definition}{Defn.}
\theoremstyle{plain}
\theoremheaderfont{\normalfont\itshape}
\theorembodyfont{\normalfont}
\theoremseparator{.}
\theoremsymbol{}

\usepackage{tikz}

\usepackage{graphicx,epstopdf}

\usepackage[font=footnotesize,justification=centering]{subfig}

\usepackage{etoolbox}
\makeatletter
\let\tightset@fontsize\set@fontsize
\patchcmd\tightset@fontsize{#3}{#2}{}{}

\makeatother

\newcommand{\cF}{\mathcal{F}}

\newcommand{\cS}{\mathcal{S}}

\newcommand{\bA}{\boldsymbol{A}}

\newcommand{\Sec}[1]{{\S\ref*{sec:#1}}} %section
\newcommand{\Fig}[1]{{Fig.\,\ref*{fig:#1}}} %figure
\newcommand{\Tab}[1]{{Tab.\,\ref*{tab:#1}}} %table
\newcommand{\Thm}[1]{{Theorem\,\ref*{thm:#1}}} %theorem
\newcommand{\Lem}[1]{{Lemma\,\ref*{lem:#1}}} %lemma
\newcommand{\ind}{C}
\newcommand{\nonind}{N}

\newcommand{\cc}{D}
\newcommand{\dcc}{DD}
\newcommand{\dpath}{DP}
\newcommand{\dbp}{DBP}
\newcommand{\tailedtri}{TT}
\newcommand{\match}{\hbox{match}}

\newcommand{\aut}{\hbox{Aut}}
\newcommand{\frag}{\hbox{Frag}}
\newcommand{\shrink}{\hbox{Shrink}}
\newcommand{\numshrink}{\hbox{numSh}}

\usepackage{tikz}
\usetikzlibrary{decorations.markings,fit,arrows,positioning,backgrounds}
\tikzset{%
  gnode/.style={shape=circle,minimum size=3mm,fill,draw=black}
}
\tikzset{myptr/.style={decoration={markings,mark=at position 1 with {\arrow[scale=2,>=stealth]{>}}},postaction={decorate}}}
\def\IndSetFour{
   \node (3) at (0,0) [nd] {3};
   \node (1) at (0,1) [nd] {1};
   \node (4) at (1,0) [nd] {4};
   \node (2) at (1,1) [nd] {2};
  \node [inner sep = 0,below] at +(0.75,-0.5) {{\Large (a) Ind set}};
}

\def\EdgeFour{
   \node (3) at (0,0) [nd] {3};
   \node (1) at (0,1) [nd] {1};
   \node (4) at (1,0) [nd] {4};
   \node (2) at (1,1) [nd] {2};
   \draw (1) to (2);
  \node [inner sep = 0,below] at +(0.75,-0.5) {{\Large (b) Only edge}};
}

\def\MatFour{
   \node (3) at (0,0) [nd] {3};
   \node (1) at (0,1) [nd] {1};
   \node (4) at (1,0) [nd] {4};
   \node (2) at (1,1) [nd] {2};
   \draw (1) to (2);
   \draw (3) to (4);
  \node [inner sep = 0,below] at +(0.75,-0.5) {{\Large (c) Matching}};
}

\def\WedgeFour{
   \node (3) at (0,0) [nd] {3};
   \node (1) at (0,1) [nd] {1};
   \node (4) at (1,0) [nd] {4};
   \node (2) at (1,1) [nd] {2};
   \draw (1) to (2);
   \draw (2) to (3);
  \node [inner sep = 0,below] at +(0.75,-0.5) {{\Large (d) Only wedge}};
}

\def\TriFour{
   \node (3) at (0,0) [nd] {3};
   \node (1) at (0,1) [nd] {1};
   \node (4) at (1,0) [nd] {4};
   \node (2) at (1,1) [nd] {2};
   \draw (1) to (2);
   \draw (1) to (3);
   \draw (2) to (3);
  \node [inner sep = 0,below] at +(0.75,-0.5) {{\Large (e) Only triangle}};
}

\def\ThreeStar {
   \node (1) at (0,0) [nd] {3};
   \node (2) at (0,1) [nd] {1};
   \node (3) at (1,0) [nd] {4};
   \node (4) at (1,1) [nd] {2};
   \draw (1) to (2);
   \draw (1) to (3);
   \draw (1) to (4);
  \node [inner sep = 0,below] at +(0.75,-0.5) {\Large (a) 3-star};
}

\def\ThreePath {
   \node (1) at (0,0) [nd] {3};
   \node (2) at (0,1) [nd] {1};
   \node (3) at (1,0) [nd] {4};
   \node (4) at (1,1) [nd] {2};
   \draw (1) to (2);
   \draw (1) to (3);
   \draw (2) to (4);
  \node [inner sep = 0,below] at +(0.75,-0.5) {\Large (b) 3-path};
}

\def\TailedTriangle {
   \node (1) at (0,0) [nd] {3};
   \node (2) at (0,1) [nd] {1};
   \node (3) at (1,0) [nd] {4};
   \node (4) at (1,1) [nd] {2};
   \draw (1) to (2);
   \draw (1) to (3);
   \draw (1) to (4);
   \draw (2) to (4);       
  \node [inner sep = 0,below] at +(0.75,-0.5) {{\Large (c) tailed triangle}};
}
\def\FourCycle {
   \node (1) at (0,0) [nd] {3};
   \node (2) at (0,1) [nd] {1};
   \node (3) at (1,0) [nd] {4};
   \node (4) at (1,1) [nd] {2};
   \draw (1) to (2);
   \draw (1) to (3);
   \draw (2) to (4);
   \draw (3) to (4);      
  \node [inner sep = 0,below] at +(0.75,-0.5) {{\Large (d) 4-cycle }};
}
\def\ChordalFourCycle {
   \node (1) at (0,0) [nd] {3};
   \node (2) at (0,1) [nd] {1};
   \node (3) at (1,0) [nd] {4};
   \node (4) at (1,1) [nd] {2};
   \draw (1) to (2);
   \draw (1) to (3);
   \draw(1) to (4);
   \draw (2) to (4);
   \draw (3) to (4);    
  \node [inner sep = 0,below] at +(0.75,-0.5) {{\Large (e) diamond}};
}
\def\FourClique  {
   \node (1) at (0,0) [nd] {3};
   \node (2) at (0,1) [nd] {1};
   \node (3) at (1,0) [nd] {4};
   \node (4) at (1,1) [nd] {2};
   \draw (1) to (2);
   \draw (1) to (3);
   \draw (1) to (4);
   \draw (2) to (3);      
   \draw (2) to (4);
   \draw (3) to (4);    
  \node [inner sep = 0,below] at +(0.75,-0.5) {{\Large (f) 4-clique}};
}

\def\FourCycleOne {
   \node (1) at (0,0) [nd] {};
   \node (2) at (0,1) [nd] {$i$};
   \node (3) at (1,0) [nd] {$j$};
   \node (4) at (1,1) [nd] {};
   \draw[-stealth'] (1) to (2);
   \draw[-stealth'] (1) to (3);
   \draw[stealth'-] (2) to (4);
   \draw[stealth'-] (3) to (4);      
  \node [inner sep = 0,below] at +(0.5,-0.5) {{\Large (a) }};
}

\def\FourCycleTwo {
   \node (1) at (0,0) [nd] {};
   \node (2) at (0,1) [nd] {$i$};
   \node (3) at (1,0) [nd] {$j$};
   \node (4) at (1,1) [nd] {};
   \draw[-stealth'] (1) to (2);
   \draw[stealth'-] (1) to (3);
   \draw[stealth'-] (2) to (4);
   \draw[-stealth'] (3) to (4);      
  \node [inner sep = 0,below] at +(0.5,-0.5) {{\Large (b) }};
}

\def\FourCycleThree {
   \node (1) at (0,0) [nd] {};
   \node (2) at (0,1) [nd] {$i$};
   \node (3) at (1,0) [nd] {$j$};
   \node (4) at (1,1) [nd] {};
   \draw[-stealth'] (1) to (2);
   \draw[stealth'-] (1) to (3);
   \draw[stealth'-] (2) to (4);
   \draw[stealth'-] (3) to (4);      
  \node [inner sep = 0,below] at +(0.5,-0.5) {{\Large (c) }};
}

\def\DirChordal{
   \node (1) at (0,0) [nd] {};
   \node (2) at (0,1) [nd] {};
   \node (3) at (1,0) [nd] {};
   \node (4) at (1,1) [nd] {};
   \draw[myptr] (1) to (2);
   \draw[myptr] (1) to (3);
   \draw[myptr] (1) to (4);
   \draw[myptr] (2) to (4);
   \draw[myptr] (3) to (4);      
  \node [inner sep = 0,below] at +(0.5,-0.5) {{\Large Directed Diamond}};
}

\def\Wedge {
   \node (1) at (0.75,1.6) [nd] {};
   \node (2) at (-0.25,2) [nd] {};
   \node (3) at (1.75,2) [nd] {};
   \draw (1) to (2);
   \draw (1) to (3);
  \node [inner sep = 0,below] at +(0.75,2.25) {{\Large Wedge}};
   \node (5) at (0.25,0) [nd] {};
   \node (6) at (0.25,1) [nd] {};
   \node (7) at (1.25,0) [nd] {};
   \node (8) at (1.25,1) [nd] {};
   \draw (5) to (6);
   \draw (5) to (7);
   \draw (5) to (8);
   \draw (6) to (8);
   \draw (7) to (8);    
  \node [inner sep = 0,below] at +(0.75,-0.5) {{\Large Diamond}};
}

\def\Outwedge {
   \node (1) at (0.75,0) [nd] {};
   \node (2) at (0,1) [nd] {};
   \node (3) at (1.5,1) [nd] {};
   \draw[myptr] (1) to (2);
   \draw[myptr] (1) to (3);
  \node [inner sep = 0,below] at +(0.75,-0.5) {{\Large Out-wedge}};
}

\def\Inoutwedge{
   \node (1) at (0.75,0) [nd] {};
   \node (2) at (0,1) [nd] {};
   \node (3) at (1.5,1) [nd] {};
   \draw[myptr] (1) to (2);
   \draw[myptr] (3) to (1);
  \node [inner sep = 0,below] at +(0.75,-0.5) {{\Large Inout-wedge}};
}

\def\pone {
   \node (1) at (0.75,1) [nd] {1};
   \node (2) at (0,2) [nd] {2};
   \node (3) at (1.5,2) [nd] {3};
   \node (4) at (0,0) [nd] {4};
   \node (5) at (1.5,0) [nd] {5};
   \draw (1) to (2);
   \draw (1) to (3);
   \draw (1) to (4);
   \draw (1) to (5);
  \node [inner sep = 0,below] at +(0.75,-0.5) {{\Large (1)}};
}

\def\ptwo {
   \node (1) at (1,2) [nd] {1};
   \node (2) at (0,1) [nd] {2};
   \node (3) at (1,1) [nd] {3};
   \node (4) at (2,1) [nd] {4};
   \node (5) at (0,0) [nd] {5};
   \draw (1) to (2);
   \draw (1) to (3);
   \draw (1) to (4);
   \draw (2) to (5);
   \node [inner sep = 0,below] at +(1,-0.5) {{\Large (2)}};
}

\def\pthree{
   \node (1) at (0.5,2) [nd] {1};
   \node (2) at (0,1) [nd] {2};
   \node (3) at (1,1) [nd] {3};
   \node (4) at (0,0) [nd] {4};
   \node (5) at (1,0) [nd] {5};
   \draw (1) to (2);
   \draw (1) to (3);
   \draw (2) to (4);
   \draw (3) to (5);
\node [inner sep = 0,below] at +(0.5,-0.5) {{\Large (3)}};
}

\def\pfour{
   \node (1) at (0,2) [nd] {1};
   \node (2) at (1,2) [nd] {2};
   \node (3) at (0.5,1) [nd] {3};
   \node (4) at (0,0) [nd] {4};
   \node (5) at (1,0) [nd] {5};
   \draw (1) to (2);
   \draw (1) to (3);
     \draw (2) to (3);
   \draw (3) to (4);
   \draw (3) to (5);
\node [inner sep = 0,below] at +(0.5,-0.5) {{\Large (4)}};
}

\def\pfive{
   \node (1) at (0.5,2) [nd] {1};
   \node (2) at (0,1) [nd] {2};
   \node (3) at (1,1) [nd] {3};
   \node (4) at (0.5,0.5) [nd] {4};
   \node (5) at (1,0) [nd] {5};
   \draw (1) to (2);
   \draw (1) to (3);
     \draw (2) to (3);	
   \draw (2) to (4);
   \draw (4) to (5);
\node [inner sep = 0,below] at +(0.5,-0.5) {{\Large (5)}};
}

\def\psix{
  \node (1) at (0.5,2) [nd] {1};
   \node (2) at (0,1) [nd] {2};
   \node (3) at (1,1) [nd] {3};
   \node (4) at (0,0) [nd] {4};
   \node (5) at (1,0) [nd] {5};
   \draw (1) to (2);
   \draw (1) to (3);
   \draw (2) to (3);
   \draw (2) to (4);
   \draw (3) to (5);
   
\node [inner sep = 0,below] at +(0.5,-0.5) {{\Large (6)}};
}

\def\pseven{
   \node (1) at (0,2) [nd] {1};
   \node (2) at (1,2) [nd] {2};
   \node (3) at (0,1) [nd] {3};
   \node (4) at (1,1) [nd] {4};
   \node (5) at (0.5,0) [nd] {5};
   \draw (1) to (2);
   \draw (2) to (4);
   \draw (3) to (4);
   \draw (3) to (1);
     \draw (3) to (5);
\node [inner sep = 0,below] at +(0.5,-0.5) {{\Large (7)}};
}
\def\peight{
   \node (1) at (0.1,2) [nd] {1};
   \node (2) at (0.9,2) [nd] {2};
   \node (3) at (-0.1,1) [nd] {3};
   \node (4) at (1.1,1) [nd] {4};
   \node (5) at (0.5,0) [nd] {5};
   \draw (1) to (2);
   \draw (1) to (3);
      \draw (4) to (2);
    	           \draw (3) to (5); 
       \draw (5) to (4);
\node [inner sep = 0,below] at +(0.5,-0.5) {{\Large (8)}};}

\def\pnine{
  \node (1) at (0,2) [nd] {1};
   \node (2) at (1,2) [nd] {2};
   \node (3) at (0.5,1) [nd] {3};
   \node (4) at (0,0) [nd] {4};
   \node (5) at (1,0) [nd] {5};
   \draw (1) to (2);
   \draw (1) to (3);
     \draw (2) to (3);
   \draw (3) to (4);
   \draw (3) to (5);
 \draw (4) to (5);
\node [inner sep = 0,below] at +(0.5,-0.5) {{\Large (9)}};
}
\def\pten{
   \node (1) at (0,2) [nd] {1};
   \node (2) at (1,2) [nd] {2};
   \node (3) at (0,1) [nd] {3};
   \node (4) at (1,1) [nd] {4};
   \node (5) at (0.5,0) [nd] {5};
   \draw (1) to (2);
   \draw (2) to (4);
   \draw (3) to (4);
   \draw (3) to (1);
     \draw (3) to (5);
      \draw (1) to (4);
\node [inner sep = 0,below] at +(0.5,-0.5) {{\Large (10)}};
}
\def\peleven{
   \node (1) at (0,2) [nd] {1};
   \node (2) at (1,2) [nd] {2};
   \node (3) at (0,1) [nd] {3};
   \node (4) at (1,1) [nd] {4};
   \node (5) at (0.5,0) [nd] {5};
   \draw (1) to (2);
   \draw (2) to (4);
   \draw (3) to (4);
   \draw (3) to (1);
     \draw (3) to (5);
      \draw (3) to (2);
\node [inner sep = 0,below] at +(0.5,-0.5) {{\Large (11)}};
}
\def\ptwelve{
   \node (1) at (0,2) [nd] {1};
   \node (2) at (1,2) [nd] {2};
   \node (3) at (0,1) [nd] {3};
   \node (4) at (1,1) [nd] {4};
   \node (5) at (0.5,0) [nd] {5};
   \draw (1) to (2);
   \draw (2) to (4);
   \draw (3) to (4);
   \draw (3) to (1);
     \draw (3) to (5);
 
      \draw (4) to (5);
\node [inner sep = 0,below] at +(0.5,-0.5) {{\Large (12)}};
}
\def\pthirteen{
   \node (1) at (0.9,2) [nd] {1};
   \node (2) at (0,1) [nd] {2};
   \node (3) at (0.5,1) [nd] {3};
   \node (4) at (1.5,1) [nd] {4};
   \node (5) at (0.9,0) [nd] {5};
   \draw (1) to (2);
   \draw (2) to (5);
   \draw (1) to (3);
   \draw (3) to (5);
\draw (1) to (4);
     \draw (4) to (5);
\node [inner sep = 0,below] at +(0.5,-0.5) {{\Large (13)}};
}
\def\pfourteen{
     \node (1) at (0.9,2) [nd] {1};
   \node (2) at (0,1) [nd] {2};
   \node (3) at (0.5,1) [nd] {3};
   \node (4) at (1.5,1) [nd] {4};
   \node (5) at (0.9,0) [nd] {5};
   \draw (1) to (2);
   \draw (2) to (5);
   \draw (1) to (3);
   \draw (3) to (5);
\draw (1) to (4);
     \draw (4) to (5);
\draw (1) to (5);
\node [inner sep = 0,below] at +(0.5,-0.5) {{\Large (14)}};
}
\def\pfifteen{
   \node (1) at (0,2) [nd] {1};
   \node (2) at (1,2) [nd] {2};
   \node (3) at (0,1) [nd] {3};
   \node (4) at (1,1) [nd] {4};
   \node (5) at (0.5,0) [nd] {5};
   \draw (1) to (2);
   \draw (2) to (4);
   \draw (3) to (4);
   \draw (3) to (1);
   \draw (3) to (5);
  \draw (1) to (4);
  \draw (3) to (2);
\node [inner sep = 0,below] at +(0.5,-0.5) {{\Large (15)}};
}
\def\psixteen{
   \node (1) at (0,2) [nd] {1};
   \node (2) at (1,2) [nd] {2};
   \node (3) at (0,1) [nd] {3};
   \node (4) at (1,1) [nd] {4};
   \node (5) at (0.5,0) [nd] {5};
   \draw (1) to (2);
   \draw (2) to (4);
   \draw (3) to (4);
   \draw (3) to (1);
     \draw (3) to (5);
      \draw (3) to (2);
      \draw (4) to (5);
\node [inner sep = 0,below] at +(0.5,-0.5) {{\Large (16)}};
}
\def\pseventeen{
    \node (1) at (0.1,2) [nd] {1};
   \node (2) at (0.9,2) [nd] {2};
   \node (3) at (-0.1,1) [nd] {3};
   \node (4) at (1.1,1) [nd] {4};
   \node (5) at (0.5,0) [nd] {5};
   \draw (1) to (2);
   \draw (1) to (3);
   \draw (1) to (4);
        \draw (3) to (2);
      \draw (4) to (2);
          \draw (3) to (5); 
       \draw (5) to (4);
\node [inner sep = 0,below] at +(0.5,-0.5) {{\Large (17)}};
}
\def\peighteen{
   \node (1) at (0.75,1) [nd] {1};
   \node (2) at (0,2) [nd] {2};
   \node (3) at (1.5,2) [nd] {3};
   \node (4) at (0,0) [nd] {4};
   \node (5) at (1.5,0) [nd] {5};
   \draw (1) to (2);
   \draw (1) to (3);
   \draw (1) to (4);
   \draw (1) to (5);
     \draw (3) to (2);
   \draw (5) to (3);
   \draw (5) to (4);
   \draw (2) to (4);

  \node [inner sep = 0,below] at +(0.75,-0.5) {{\Large (18)}};
}
\def\pnineteen{
  \node (1) at (0,2) [nd] {1};
   \node (2) at (1,2) [nd] {2};
   \node (3) at (0,1) [nd] {3};
   \node (4) at (1,1) [nd] {4};
   \node (5) at (0.5,0) [nd] {5};
   \draw (1) to (2);
   \draw (2) to (4);
   \draw (3) to (4);
   \draw (3) to (1);
     \draw (3) to (5);
      \draw (3) to (2);
      \draw (4) to (5);
      \draw (1) to (4);
\node [inner sep = 0,below] at +(0.5,-0.5) {{\Large (19)}};

}
\def\ptwenty{
  \node (1) at (0.1,2) [nd] {1};
   \node (2) at (0.9,2) [nd] {2};
   \node (3) at (-0.1,1) [nd] {3};
   \node (4) at (1.1,1) [nd] {4};
   \node (5) at (0.5,0) [nd] {5};
   \draw (1) to (3);
   \draw (1) to (4);
   \draw (1) to (5);
     \draw (3) to (2);
      \draw (4) to (2);
      \draw (2) to (5);
      \draw (3) to (4);
       \draw (3) to (5); 
    \draw (4) to (5);
\node [inner sep = 0,below] at +(0.5,-0.5) {{\Large (20)}};
}

\def\ptwentyone{
   \node (1) at (0.1,2) [nd] {1};
   \node (2) at (0.9,2) [nd] {2};
   \node (3) at (-0.1,1) [nd] {3};
   \node (4) at (1.1,1) [nd] {4};
   \node (5) at (0.5,0) [nd] {5};
   \draw (1) to (2);
   \draw (1) to (3);
   \draw (1) to (4);
   \draw (1) to (5);
     \draw (3) to (2);
      \draw (4) to (2);
      \draw (2) to (5);
      \draw (3) to (4);
       \draw (3) to (5); 
       \draw (5) to (4);
\node [inner sep = 0,below] at +(0.5,-0.5) {{\Large (21)}};
}

\def\IndSetFive{
   \node (1) at (0.1,2) [nd] {1};
   \node (2) at (0.9,2) [nd] {2};
   \node (3) at (-0.1,1) [nd] {3};
   \node (4) at (1.1,1) [nd] {4};
   \node (5) at (0.5,0) [nd] {5};
\node [inner sep = 0,below] at +(0.5,-0.5) {{\Large (1) Ind Set}};
}

\def\EdgeFive{
   \node (1) at (0.1,2) [nd] {1};
   \node (2) at (0.9,2) [nd] {2};
   \node (3) at (-0.1,1) [nd] {3};
   \node (4) at (1.1,1) [nd] {4};
   \node (5) at (0.5,0) [nd] {5};
   \draw (1) to (2);
\node [inner sep = 0,below] at +(0.5,-0.5) {{\Large (2) Only Edge}};
}

\def\MatFive{
   \node (1) at (0.1,2) [nd] {1};
   \node (2) at (0.9,2) [nd] {2};
   \node (3) at (-0.1,1) [nd] {3};
   \node (4) at (1.1,1) [nd] {4};
   \node (5) at (0.5,0) [nd] {5};
   \draw (1) to (2);
   \draw (3) to (4);
\node [inner sep = 0,below] at +(0.5,-0.5) {{\Large (3) Matching}};
}

\def\WedgeFive{
   \node (1) at (0.1,2) [nd] {1};
   \node (2) at (0.9,2) [nd] {2};
   \node (3) at (-0.1,1) [nd] {3};
   \node (4) at (1.1,1) [nd] {4};
   \node (5) at (0.5,0) [nd] {5};
   \draw (1) to (2);
   \draw (1) to (3);
\node [inner sep = 0,below] at +(0.5,-0.5) {{\Large (4) Only Wedge}};
}

\def\TriFive{
   \node (1) at (0.1,2) [nd] {1};
   \node (2) at (0.9,2) [nd] {2};
   \node (3) at (-0.1,1) [nd] {3};
   \node (4) at (1.1,1) [nd] {4};
   \node (5) at (0.5,0) [nd] {5};
   \draw (1) to (2);
   \draw (1) to (3);
   \draw (2) to (3);
\node [inner sep = 0,below] at +(0.5,-0.5) {{\Large (5) Only Triangle}};
}

\def\ThreeStarFive{
   \node (1) at (0.1,2) [nd] {1};
   \node (2) at (0.9,2) [nd] {2};
   \node (3) at (-0.1,1) [nd] {3};
   \node (4) at (1.1,1) [nd] {4};
   \node (5) at (0.5,0) [nd] {5};
   \draw (1) to (2);
   \draw (1) to (3);
   \draw (1) to (4);
\node [inner sep = 0,below] at +(0.5,-0.5) {{\Large (6) Only 3-Star}};
}

\def\ThreePathFive{
   \node (1) at (0.1,2) [nd] {1};
   \node (2) at (0.9,2) [nd] {2};
   \node (3) at (-0.1,1) [nd] {3};
   \node (4) at (1.1,1) [nd] {4};
   \node (5) at (0.5,0) [nd] {5};
   \draw (1) to (2);
   \draw (1) to (3);
   \draw (3) to (4);
\node [inner sep = 0,below] at +(0.5,-0.5) {{\Large (7) Only 3-Path}};
}

\def\TailedTriFive{
   \node (1) at (0.1,2) [nd] {1};
   \node (2) at (0.9,2) [nd] {2};
   \node (3) at (-0.1,1) [nd] {3};
   \node (4) at (1.1,1) [nd] {4};
   \node (5) at (0.5,0) [nd] {5};
   \draw (1) to (2);
   \draw (2) to (3);
   \draw (1) to (3);
   \draw (3) to (4);
\node [inner sep = 0,below] at +(0.5,-0.5) {{\Large (8) Only Tailed-Triangle}};
}

\def\FourCycleFive{
   \node (1) at (0.1,2) [nd] {1};
   \node (2) at (0.9,2) [nd] {2};
   \node (3) at (-0.1,1) [nd] {3};
   \node (4) at (1.1,1) [nd] {4};
   \node (5) at (0.5,0) [nd] {5};
   \draw (1) to (2);
   \draw (2) to (4);
   \draw (3) to (4);
   \draw (1) to (3);
\node [inner sep = 0,below] at +(0.5,-0.5) {{\Large (9) Only 4-cycle}};
}

\def\DiamondFive{
   \node (1) at (0.1,2) [nd] {1};
   \node (2) at (0.9,2) [nd] {2};
   \node (3) at (-0.1,1) [nd] {3};
   \node (4) at (1.1,1) [nd] {4};
   \node (5) at (0.5,0) [nd] {5};
   \draw (1) to (2);
   \draw (2) to (4);
   \draw (3) to (4);
   \draw (1) to (4);
   \draw (1) to (3);
\node [inner sep = 0,below] at +(0.5,-0.5) {{\Large (10) Only Diamond}};
}

\def\FourCliqueFive{
   \node (1) at (0.1,2) [nd] {1};
   \node (2) at (0.9,2) [nd] {2};
   \node (3) at (-0.1,1) [nd] {3};
   \node (4) at (1.1,1) [nd] {4};
   \node (5) at (0.5,0) [nd] {5};
   \draw (1) to (2);
   \draw (2) to (3);
   \draw (3) to (4);
   \draw (1) to (4);
   \draw (2) to (4);
   \draw (1) to (3);
\node [inner sep = 0,below] at +(0.5,-0.5) {{\Large (11) Only 4-clique}};
}

\def\FiveCycleOne{
   \node (1) at (0.1,2) [nd, fill=white] {$j$};
   \node (2) at (0.9,2) [nd] { };
   \node (3) at (-0.1,1) [nd] { };
   \node (4) at (1.1,1) [nd, fill=white] {$i$};
   \node (5) at (0.5,0) [nd] {};
   \draw[-stealth'] (5) to (3);
   \draw[-stealth'] (5) to (4);
   \draw[-stealth'] (3) to (1);
   \draw[-stealth'] (4) to (2);
   \draw[-stealth'] (2) to (1);
\node [inner sep = 0,below] at +(0.5,-0.5) {{\Large (a)}};
}

\def\FiveCycleTwo{
   \node (1) at (0.1,2) [nd, fill=white] {$j$};
   \node (2) at (0.9,2) [nd] { };
   \node (3) at (-0.1,1) [nd] { };
   \node (4) at (1.1,1) [nd, fill=white] {$i$};
   \node (5) at (0.5,0) [nd] {};
   \draw[-stealth'] (5) to (3);
   \draw[-stealth'] (5) to (4);
   \draw[-stealth'] (3) to (1);
   \draw[stealth'-] (4) to (2);
   \draw[-stealth'] (2) to (1);
\node [inner sep = 0,below] at +(0.5,-0.5) {{\Large (b)}};
}

\def\DirThreePath{
   \node (1) at (0.1,2) [nd] {};
   \node (3) at (-0.1,1) [nd] { };
   \node[fill=gray] (4) at (1.1,1) [nd] {};
   \node (5) at (0.5,0) [nd] { };
   \draw[myptr] (5) to (3);
   \draw[myptr] (5) to (4);
   \draw[myptr] (3) to (1);
\node [inner sep = 0,below] at +(0.5,-0.5) {{\Large Directed 3-path}};
}

\def\DirTailedTri{
   \node (1) at (0.1,2) [nd] {$i$};
   \node (3) at (-0.1,1) [nd] { };
   \node[fill=gray] (4) at (1.1,1) [nd] {$j$};
   \node (5) at (0.5,0) [nd] { };
   \draw[myptr] (5) to (3);
   \draw[myptr] (5) to (4);
   \draw[myptr] (4) to (3);
   \draw[myptr] (3) to (1);
\node [inner sep = 0,below] at +(0.5,-0.5) {{\Large Directed tailed-triangle}};
}

\def\DirForClique{
  \node (1) at (0.1,2) [nd] {};
   \node (2) at (0.9,2) [nd] {$\ell$};
   \node (3) at (-0.1,1) [nd ] {$k$};
   \node (4) at (1.1,1) [nd] {$j$};
   \node (5) at (0.5,0) [nd] {$i$};
   \draw[myptr] (3) to (1);
   \draw[myptr] (4) to (1);
   \draw[myptr] (4) to (1);
   \draw[myptr] (3) to (2);
   \draw[myptr] (2) to (4);
   \draw[myptr] (5) to (1);
   \draw[myptr] (5) to (2);
   \draw[myptr] (4) to (3);
   \draw[myptr] (5) to (3); 
   \draw[myptr] (5) to (4);
\node [inner sep = 0,below] at +(0.5,-0.5) {{\Large Directed bipyramid}};
}

\usepackage{tikz}
\usetikzlibrary{decorations.markings,fit,arrows,positioning,backgrounds}
\tikzset{%
  gnode/.style={shape=circle,minimum size=3mm,fill,draw=black}
}
\tikzset{myptr/.style={decoration={markings,mark=at position 1 with {\arrow[scale=2,>=stealth]{>}}},postaction={decorate}}}
\def\IndSetFour{
   \node (3) at (0,0) [nd] {3};
   \node (1) at (0,1) [nd] {1};
   \node (4) at (1,0) [nd] {4};
   \node (2) at (1,1) [nd] {2};
  \node [inner sep = 0,below] at +(0.75,-0.5) {{\Large (a) Ind set}};
}

\def\EdgeFour{
   \node (3) at (0,0) [nd] {3};
   \node (1) at (0,1) [nd] {1};
   \node (4) at (1,0) [nd] {4};
   \node (2) at (1,1) [nd] {2};
   \draw (1) to (2);
  \node [inner sep = 0,below] at +(0.75,-0.5) {{\Large (b) Only edge}};
}

\def\MatFour{
   \node (3) at (0,0) [nd] {3};
   \node (1) at (0,1) [nd] {1};
   \node (4) at (1,0) [nd] {4};
   \node (2) at (1,1) [nd] {2};
   \draw (1) to (2);
   \draw (3) to (4);
  \node [inner sep = 0,below] at +(0.75,-0.5) {{\Large (c) Matching}};
}

\def\WedgeFour{
   \node (3) at (0,0) [nd] {3};
   \node (1) at (0,1) [nd] {1};
   \node (4) at (1,0) [nd] {4};
   \node (2) at (1,1) [nd] {2};
   \draw (1) to (2);
   \draw (2) to (3);
  \node [inner sep = 0,below] at +(0.75,-0.5) {{\Large (d) Only wedge}};
}

\def\TriFour{
   \node (3) at (0,0) [nd] {3};
   \node (1) at (0,1) [nd] {1};
   \node (4) at (1,0) [nd] {4};
   \node (2) at (1,1) [nd] {2};
   \draw (1) to (2);
   \draw (1) to (3);
   \draw (2) to (3);
  \node [inner sep = 0,below] at +(0.75,-0.5) {{\Large (e) Only triangle}};
}

\def\ThreeStar {
   \node (1) at (0,0) [nd] {3};
   \node (2) at (0,1) [nd] {1};
   \node (3) at (1,0) [nd] {4};
   \node (4) at (1,1) [nd] {2};
   \draw (1) to (2);
   \draw (1) to (3);
   \draw (1) to (4);
  \node [inner sep = 0,below] at +(0.75,-0.5) {\Large (a) 3-star};
}

\def\ThreePath {
   \node (1) at (0,0) [nd] {3};
   \node (2) at (0,1) [nd] {1};
   \node (3) at (1,0) [nd] {4};
   \node (4) at (1,1) [nd] {2};
   \draw (1) to (2);
   \draw (1) to (3);
   \draw (2) to (4);
  \node [inner sep = 0,below] at +(0.75,-0.5) {\Large (b) 3-path};
}

\def\TailedTriangle {
   \node (1) at (0,0) [nd] {3};
   \node (2) at (0,1) [nd] {1};
   \node (3) at (1,0) [nd] {4};
   \node (4) at (1,1) [nd] {2};
   \draw (1) to (2);
   \draw (1) to (3);
   \draw (1) to (4);
   \draw (2) to (4);       
  \node [inner sep = 0,below] at +(0.75,-0.5) {{\Large (c) tailed triangle}};
}
\def\FourCycle {
   \node (1) at (0,0) [nd] {3};
   \node (2) at (0,1) [nd] {1};
   \node (3) at (1,0) [nd] {4};
   \node (4) at (1,1) [nd] {2};
   \draw (1) to (2);
   \draw (1) to (3);
   \draw (2) to (4);
   \draw (3) to (4);      
  \node [inner sep = 0,below] at +(0.75,-0.5) {{\Large (d) 4-cycle }};
}
\def\ChordalFourCycle {
   \node (1) at (0,0) [nd] {3};
   \node (2) at (0,1) [nd] {1};
   \node (3) at (1,0) [nd] {4};
   \node (4) at (1,1) [nd] {2};
   \draw (1) to (2);
   \draw (1) to (3);
   \draw(1) to (4);
   \draw (2) to (4);
   \draw (3) to (4);    
  \node [inner sep = 0,below] at +(0.75,-0.5) {{\Large (e) diamond}};
}
\def\FourClique  {
   \node (1) at (0,0) [nd] {3};
   \node (2) at (0,1) [nd] {1};
   \node (3) at (1,0) [nd] {4};
   \node (4) at (1,1) [nd] {2};
   \draw (1) to (2);
   \draw (1) to (3);
   \draw (1) to (4);
   \draw (2) to (3);      
   \draw (2) to (4);
   \draw (3) to (4);    
  \node [inner sep = 0,below] at +(0.75,-0.5) {{\Large (f) 4-clique}};
}

\def\FourCycleOne {
   \node (1) at (0,0) [nd] {};
   \node (2) at (0,1) [nd] {$i$};
   \node (3) at (1,0) [nd] {$j$};
   \node (4) at (1,1) [nd] {};
   \draw[-stealth'] (1) to (2);
   \draw[-stealth'] (1) to (3);
   \draw[stealth'-] (2) to (4);
   \draw[stealth'-] (3) to (4);      
  \node [inner sep = 0,below] at +(0.5,-0.5) {{\Large (a) }};
}

\def\FourCycleTwo {
   \node (1) at (0,0) [nd] {};
   \node (2) at (0,1) [nd] {$i$};
   \node (3) at (1,0) [nd] {$j$};
   \node (4) at (1,1) [nd] {};
   \draw[-stealth'] (1) to (2);
   \draw[stealth'-] (1) to (3);
   \draw[stealth'-] (2) to (4);
   \draw[-stealth'] (3) to (4);      
  \node [inner sep = 0,below] at +(0.5,-0.5) {{\Large (b) }};
}

\def\FourCycleThree {
   \node (1) at (0,0) [nd] {};
   \node (2) at (0,1) [nd] {$i$};
   \node (3) at (1,0) [nd] {$j$};
   \node (4) at (1,1) [nd] {};
   \draw[-stealth'] (1) to (2);
   \draw[stealth'-] (1) to (3);
   \draw[stealth'-] (2) to (4);
   \draw[stealth'-] (3) to (4);      
  \node [inner sep = 0,below] at +(0.5,-0.5) {{\Large (c) }};
}

\def\DirChordal{
   \node (1) at (0,0) [nd] {};
   \node (2) at (0,1) [nd] {};
   \node (3) at (1,0) [nd] {};
   \node (4) at (1,1) [nd] {};
   \draw[myptr] (1) to (2);
   \draw[myptr] (1) to (3);
   \draw[myptr] (1) to (4);
   \draw[myptr] (2) to (4);
   \draw[myptr] (3) to (4);      
  \node [inner sep = 0,below] at +(0.5,-0.5) {{\Large Directed Diamond}};
}

\def\Wedge {
   \node (1) at (0.75,1.6) [nd] {};
   \node (2) at (-0.25,2) [nd] {};
   \node (3) at (1.75,2) [nd] {};
   \draw (1) to (2);
   \draw (1) to (3);
  \node [inner sep = 0,below] at +(0.75,2.25) {{\Large Wedge}};
   \node (5) at (0.25,0) [nd] {};
   \node (6) at (0.25,1) [nd] {};
   \node (7) at (1.25,0) [nd] {};
   \node (8) at (1.25,1) [nd] {};
   \draw (5) to (6);
   \draw (5) to (7);
   \draw (5) to (8);
   \draw (6) to (8);
   \draw (7) to (8);    
  \node [inner sep = 0,below] at +(0.75,-0.5) {{\Large Diamond}};
}

\def\Outwedge {
   \node (1) at (0.75,0) [nd] {};
   \node (2) at (0,1) [nd] {};
   \node (3) at (1.5,1) [nd] {};
   \draw[myptr] (1) to (2);
   \draw[myptr] (1) to (3);
  \node [inner sep = 0,below] at +(0.75,-0.5) {{\Large Out-wedge}};
}

\def\Inoutwedge{
   \node (1) at (0.75,0) [nd] {};
   \node (2) at (0,1) [nd] {};
   \node (3) at (1.5,1) [nd] {};
   \draw[myptr] (1) to (2);
   \draw[myptr] (3) to (1);
  \node [inner sep = 0,below] at +(0.75,-0.5) {{\Large Inout-wedge}};
}

\def\pone {
   \node (1) at (0.75,1) [nd] {1};
   \node (2) at (0,2) [nd] {2};
   \node (3) at (1.5,2) [nd] {3};
   \node (4) at (0,0) [nd] {4};
   \node (5) at (1.5,0) [nd] {5};
   \draw (1) to (2);
   \draw (1) to (3);
   \draw (1) to (4);
   \draw (1) to (5);
  \node [inner sep = 0,below] at +(0.75,-0.5) {{\Large (1)}};
}

\def\ptwo {
   \node (1) at (1,2) [nd] {1};
   \node (2) at (0,1) [nd] {2};
   \node (3) at (1,1) [nd] {3};
   \node (4) at (2,1) [nd] {4};
   \node (5) at (0,0) [nd] {5};
   \draw (1) to (2);
   \draw (1) to (3);
   \draw (1) to (4);
   \draw (2) to (5);
   \node [inner sep = 0,below] at +(1,-0.5) {{\Large (2)}};
}

\def\pthree{
   \node (1) at (0.5,2) [nd] {1};
   \node (2) at (0,1) [nd] {2};
   \node (3) at (1,1) [nd] {3};
   \node (4) at (0,0) [nd] {4};
   \node (5) at (1,0) [nd] {5};
   \draw (1) to (2);
   \draw (1) to (3);
   \draw (2) to (4);
   \draw (3) to (5);
\node [inner sep = 0,below] at +(0.5,-0.5) {{\Large (3)}};
}

\def\pfour{
   \node (1) at (0,2) [nd] {1};
   \node (2) at (1,2) [nd] {2};
   \node (3) at (0.5,1) [nd] {3};
   \node (4) at (0,0) [nd] {4};
   \node (5) at (1,0) [nd] {5};
   \draw (1) to (2);
   \draw (1) to (3);
     \draw (2) to (3);
   \draw (3) to (4);
   \draw (3) to (5);
\node [inner sep = 0,below] at +(0.5,-0.5) {{\Large (4)}};
}

\def\pfive{
   \node (1) at (0.5,2) [nd] {1};
   \node (2) at (0,1) [nd] {2};
   \node (3) at (1,1) [nd] {3};
   \node (4) at (0.5,0.5) [nd] {4};
   \node (5) at (1,0) [nd] {5};
   \draw (1) to (2);
   \draw (1) to (3);
     \draw (2) to (3);	
   \draw (2) to (4);
   \draw (4) to (5);
\node [inner sep = 0,below] at +(0.5,-0.5) {{\Large (5)}};
}

\def\psix{
  \node (1) at (0.5,2) [nd] {1};
   \node (2) at (0,1) [nd] {2};
   \node (3) at (1,1) [nd] {3};
   \node (4) at (0,0) [nd] {4};
   \node (5) at (1,0) [nd] {5};
   \draw (1) to (2);
   \draw (1) to (3);
   \draw (2) to (3);
   \draw (2) to (4);
   \draw (3) to (5);
   
\node [inner sep = 0,below] at +(0.5,-0.5) {{\Large (6)}};
}

\def\pseven{
   \node (1) at (0,2) [nd] {1};
   \node (2) at (1,2) [nd] {2};
   \node (3) at (0,1) [nd] {3};
   \node (4) at (1,1) [nd] {4};
   \node (5) at (0.5,0) [nd] {5};
   \draw (1) to (2);
   \draw (2) to (4);
   \draw (3) to (4);
   \draw (3) to (1);
     \draw (3) to (5);
\node [inner sep = 0,below] at +(0.5,-0.5) {{\Large (7)}};
}
\def\peight{
   \node (1) at (0.1,2) [nd] {1};
   \node (2) at (0.9,2) [nd] {2};
   \node (3) at (-0.1,1) [nd] {3};
   \node (4) at (1.1,1) [nd] {4};
   \node (5) at (0.5,0) [nd] {5};
   \draw (1) to (2);
   \draw (1) to (3);
      \draw (4) to (2);
    	           \draw (3) to (5); 
       \draw (5) to (4);
\node [inner sep = 0,below] at +(0.5,-0.5) {{\Large (8)}};}

\def\pnine{
  \node (1) at (0,2) [nd] {1};
   \node (2) at (1,2) [nd] {2};
   \node (3) at (0.5,1) [nd] {3};
   \node (4) at (0,0) [nd] {4};
   \node (5) at (1,0) [nd] {5};
   \draw (1) to (2);
   \draw (1) to (3);
     \draw (2) to (3);
   \draw (3) to (4);
   \draw (3) to (5);
 \draw (4) to (5);
\node [inner sep = 0,below] at +(0.5,-0.5) {{\Large (9)}};
}
\def\pten{
   \node (1) at (0,2) [nd] {1};
   \node (2) at (1,2) [nd] {2};
   \node (3) at (0,1) [nd] {3};
   \node (4) at (1,1) [nd] {4};
   \node (5) at (0.5,0) [nd] {5};
   \draw (1) to (2);
   \draw (2) to (4);
   \draw (3) to (4);
   \draw (3) to (1);
     \draw (3) to (5);
      \draw (1) to (4);
\node [inner sep = 0,below] at +(0.5,-0.5) {{\Large (10)}};
}
\def\peleven{
   \node (1) at (0,2) [nd] {1};
   \node (2) at (1,2) [nd] {2};
   \node (3) at (0,1) [nd] {3};
   \node (4) at (1,1) [nd] {4};
   \node (5) at (0.5,0) [nd] {5};
   \draw (1) to (2);
   \draw (2) to (4);
   \draw (3) to (4);
   \draw (3) to (1);
     \draw (3) to (5);
      \draw (3) to (2);
\node [inner sep = 0,below] at +(0.5,-0.5) {{\Large (11)}};
}
\def\ptwelve{
   \node (1) at (0,2) [nd] {1};
   \node (2) at (1,2) [nd] {2};
   \node (3) at (0,1) [nd] {3};
   \node (4) at (1,1) [nd] {4};
   \node (5) at (0.5,0) [nd] {5};
   \draw (1) to (2);
   \draw (2) to (4);
   \draw (3) to (4);
   \draw (3) to (1);
     \draw (3) to (5);
 
      \draw (4) to (5);
\node [inner sep = 0,below] at +(0.5,-0.5) {{\Large (12)}};
}
\def\pthirteen{
   \node (1) at (0.9,2) [nd] {1};
   \node (2) at (0,1) [nd] {2};
   \node (3) at (0.5,1) [nd] {3};
   \node (4) at (1.5,1) [nd] {4};
   \node (5) at (0.9,0) [nd] {5};
   \draw (1) to (2);
   \draw (2) to (5);
   \draw (1) to (3);
   \draw (3) to (5);
\draw (1) to (4);
     \draw (4) to (5);
\node [inner sep = 0,below] at +(0.5,-0.5) {{\Large (13)}};
}
\def\pfourteen{
     \node (1) at (0.9,2) [nd] {1};
   \node (2) at (0,1) [nd] {2};
   \node (3) at (0.5,1) [nd] {3};
   \node (4) at (1.5,1) [nd] {4};
   \node (5) at (0.9,0) [nd] {5};
   \draw (1) to (2);
   \draw (2) to (5);
   \draw (1) to (3);
   \draw (3) to (5);
\draw (1) to (4);
     \draw (4) to (5);
\draw (1) to (5);
\node [inner sep = 0,below] at +(0.5,-0.5) {{\Large (14)}};
}
\def\pfifteen{
   \node (1) at (0,2) [nd] {1};
   \node (2) at (1,2) [nd] {2};
   \node (3) at (0,1) [nd] {3};
   \node (4) at (1,1) [nd] {4};
   \node (5) at (0.5,0) [nd] {5};
   \draw (1) to (2);
   \draw (2) to (4);
   \draw (3) to (4);
   \draw (3) to (1);
   \draw (3) to (5);
  \draw (1) to (4);
  \draw (3) to (2);
\node [inner sep = 0,below] at +(0.5,-0.5) {{\Large (15)}};
}
\def\psixteen{
   \node (1) at (0,2) [nd] {1};
   \node (2) at (1,2) [nd] {2};
   \node (3) at (0,1) [nd] {3};
   \node (4) at (1,1) [nd] {4};
   \node (5) at (0.5,0) [nd] {5};
   \draw (1) to (2);
   \draw (2) to (4);
   \draw (3) to (4);
   \draw (3) to (1);
     \draw (3) to (5);
      \draw (3) to (2);
      \draw (4) to (5);
\node [inner sep = 0,below] at +(0.5,-0.5) {{\Large (16)}};
}
\def\pseventeen{
    \node (1) at (0.1,2) [nd] {1};
   \node (2) at (0.9,2) [nd] {2};
   \node (3) at (-0.1,1) [nd] {3};
   \node (4) at (1.1,1) [nd] {4};
   \node (5) at (0.5,0) [nd] {5};
   \draw (1) to (2);
   \draw (1) to (3);
   \draw (1) to (4);
        \draw (3) to (2);
      \draw (4) to (2);
          \draw (3) to (5); 
       \draw (5) to (4);
\node [inner sep = 0,below] at +(0.5,-0.5) {{\Large (17)}};
}
\def\peighteen{
   \node (1) at (0.75,1) [nd] {1};
   \node (2) at (0,2) [nd] {2};
   \node (3) at (1.5,2) [nd] {3};
   \node (4) at (0,0) [nd] {4};
   \node (5) at (1.5,0) [nd] {5};
   \draw (1) to (2);
   \draw (1) to (3);
   \draw (1) to (4);
   \draw (1) to (5);
     \draw (3) to (2);
   \draw (5) to (3);
   \draw (5) to (4);
   \draw (2) to (4);

  \node [inner sep = 0,below] at +(0.75,-0.5) {{\Large (18)}};
}
\def\pnineteen{
  \node (1) at (0,2) [nd] {1};
   \node (2) at (1,2) [nd] {2};
   \node (3) at (0,1) [nd] {3};
   \node (4) at (1,1) [nd] {4};
   \node (5) at (0.5,0) [nd] {5};
   \draw (1) to (2);
   \draw (2) to (4);
   \draw (3) to (4);
   \draw (3) to (1);
     \draw (3) to (5);
      \draw (3) to (2);
      \draw (4) to (5);
      \draw (1) to (4);
\node [inner sep = 0,below] at +(0.5,-0.5) {{\Large (19)}};

}
\def\ptwenty{
  \node (1) at (0.1,2) [nd] {1};
   \node (2) at (0.9,2) [nd] {2};
   \node (3) at (-0.1,1) [nd] {3};
   \node (4) at (1.1,1) [nd] {4};
   \node (5) at (0.5,0) [nd] {5};
   \draw (1) to (3);
   \draw (1) to (4);
   \draw (1) to (5);
     \draw (3) to (2);
      \draw (4) to (2);
      \draw (2) to (5);
      \draw (3) to (4);
       \draw (3) to (5); 
    \draw (4) to (5);
\node [inner sep = 0,below] at +(0.5,-0.5) {{\Large (20)}};
}

\def\ptwentyone{
   \node (1) at (0.1,2) [nd] {1};
   \node (2) at (0.9,2) [nd] {2};
   \node (3) at (-0.1,1) [nd] {3};
   \node (4) at (1.1,1) [nd] {4};
   \node (5) at (0.5,0) [nd] {5};
   \draw (1) to (2);
   \draw (1) to (3);
   \draw (1) to (4);
   \draw (1) to (5);
     \draw (3) to (2);
      \draw (4) to (2);
      \draw (2) to (5);
      \draw (3) to (4);
       \draw (3) to (5); 
       \draw (5) to (4);
\node [inner sep = 0,below] at +(0.5,-0.5) {{\Large (21)}};
}

\def\IndSetFive{
   \node (1) at (0.1,2) [nd] {1};
   \node (2) at (0.9,2) [nd] {2};
   \node (3) at (-0.1,1) [nd] {3};
   \node (4) at (1.1,1) [nd] {4};
   \node (5) at (0.5,0) [nd] {5};
\node [inner sep = 0,below] at +(0.5,-0.5) {{\Large (1) Ind Set}};
}

\def\EdgeFive{
   \node (1) at (0.1,2) [nd] {1};
   \node (2) at (0.9,2) [nd] {2};
   \node (3) at (-0.1,1) [nd] {3};
   \node (4) at (1.1,1) [nd] {4};
   \node (5) at (0.5,0) [nd] {5};
   \draw (1) to (2);
\node [inner sep = 0,below] at +(0.5,-0.5) {{\Large (2) Only Edge}};
}

\def\MatFive{
   \node (1) at (0.1,2) [nd] {1};
   \node (2) at (0.9,2) [nd] {2};
   \node (3) at (-0.1,1) [nd] {3};
   \node (4) at (1.1,1) [nd] {4};
   \node (5) at (0.5,0) [nd] {5};
   \draw (1) to (2);
   \draw (3) to (4);
\node [inner sep = 0,below] at +(0.5,-0.5) {{\Large (3) Matching}};
}

\def\WedgeFive{
   \node (1) at (0.1,2) [nd] {1};
   \node (2) at (0.9,2) [nd] {2};
   \node (3) at (-0.1,1) [nd] {3};
   \node (4) at (1.1,1) [nd] {4};
   \node (5) at (0.5,0) [nd] {5};
   \draw (1) to (2);
   \draw (1) to (3);
\node [inner sep = 0,below] at +(0.5,-0.5) {{\Large (4) Only Wedge}};
}

\def\TriFive{
   \node (1) at (0.1,2) [nd] {1};
   \node (2) at (0.9,2) [nd] {2};
   \node (3) at (-0.1,1) [nd] {3};
   \node (4) at (1.1,1) [nd] {4};
   \node (5) at (0.5,0) [nd] {5};
   \draw (1) to (2);
   \draw (1) to (3);
   \draw (2) to (3);
\node [inner sep = 0,below] at +(0.5,-0.5) {{\Large (5) Only Triangle}};
}

\def\ThreeStarFive{
   \node (1) at (0.1,2) [nd] {1};
   \node (2) at (0.9,2) [nd] {2};
   \node (3) at (-0.1,1) [nd] {3};
   \node (4) at (1.1,1) [nd] {4};
   \node (5) at (0.5,0) [nd] {5};
   \draw (1) to (2);
   \draw (1) to (3);
   \draw (1) to (4);
\node [inner sep = 0,below] at +(0.5,-0.5) {{\Large (6) Only 3-Star}};
}

\def\ThreePathFive{
   \node (1) at (0.1,2) [nd] {1};
   \node (2) at (0.9,2) [nd] {2};
   \node (3) at (-0.1,1) [nd] {3};
   \node (4) at (1.1,1) [nd] {4};
   \node (5) at (0.5,0) [nd] {5};
   \draw (1) to (2);
   \draw (1) to (3);
   \draw (3) to (4);
\node [inner sep = 0,below] at +(0.5,-0.5) {{\Large (7) Only 3-Path}};
}

\def\TailedTriFive{
   \node (1) at (0.1,2) [nd] {1};
   \node (2) at (0.9,2) [nd] {2};
   \node (3) at (-0.1,1) [nd] {3};
   \node (4) at (1.1,1) [nd] {4};
   \node (5) at (0.5,0) [nd] {5};
   \draw (1) to (2);
   \draw (2) to (3);
   \draw (1) to (3);
   \draw (3) to (4);
\node [inner sep = 0,below] at +(0.5,-0.5) {{\Large (8) Only Tailed-Triangle}};
}

\def\FourCycleFive{
   \node (1) at (0.1,2) [nd] {1};
   \node (2) at (0.9,2) [nd] {2};
   \node (3) at (-0.1,1) [nd] {3};
   \node (4) at (1.1,1) [nd] {4};
   \node (5) at (0.5,0) [nd] {5};
   \draw (1) to (2);
   \draw (2) to (4);
   \draw (3) to (4);
   \draw (1) to (3);
\node [inner sep = 0,below] at +(0.5,-0.5) {{\Large (9) Only 4-cycle}};
}

\def\DiamondFive{
   \node (1) at (0.1,2) [nd] {1};
   \node (2) at (0.9,2) [nd] {2};
   \node (3) at (-0.1,1) [nd] {3};
   \node (4) at (1.1,1) [nd] {4};
   \node (5) at (0.5,0) [nd] {5};
   \draw (1) to (2);
   \draw (2) to (4);
   \draw (3) to (4);
   \draw (1) to (4);
   \draw (1) to (3);
\node [inner sep = 0,below] at +(0.5,-0.5) {{\Large (10) Only Diamond}};
}

\def\FourCliqueFive{
   \node (1) at (0.1,2) [nd] {1};
   \node (2) at (0.9,2) [nd] {2};
   \node (3) at (-0.1,1) [nd] {3};
   \node (4) at (1.1,1) [nd] {4};
   \node (5) at (0.5,0) [nd] {5};
   \draw (1) to (2);
   \draw (2) to (3);
   \draw (3) to (4);
   \draw (1) to (4);
   \draw (2) to (4);
   \draw (1) to (3);
\node [inner sep = 0,below] at +(0.5,-0.5) {{\Large (11) Only 4-clique}};
}

\def\FiveCycleOne{
   \node (1) at (0.1,2) [nd, fill=white] {$j$};
   \node (2) at (0.9,2) [nd] { };
   \node (3) at (-0.1,1) [nd] { };
   \node (4) at (1.1,1) [nd, fill=white] {$i$};
   \node (5) at (0.5,0) [nd] {};
   \draw[-stealth'] (5) to (3);
   \draw[-stealth'] (5) to (4);
   \draw[-stealth'] (3) to (1);
   \draw[-stealth'] (4) to (2);
   \draw[-stealth'] (2) to (1);
\node [inner sep = 0,below] at +(0.5,-0.5) {{\Large (a)}};
}

\def\FiveCycleTwo{
   \node (1) at (0.1,2) [nd, fill=white] {$j$};
   \node (2) at (0.9,2) [nd] { };
   \node (3) at (-0.1,1) [nd] { };
   \node (4) at (1.1,1) [nd, fill=white] {$i$};
   \node (5) at (0.5,0) [nd] {};
   \draw[-stealth'] (5) to (3);
   \draw[-stealth'] (5) to (4);
   \draw[-stealth'] (3) to (1);
   \draw[stealth'-] (4) to (2);
   \draw[-stealth'] (2) to (1);
\node [inner sep = 0,below] at +(0.5,-0.5) {{\Large (b)}};
}

\def\DirThreePath{
   \node (1) at (0.1,2) [nd] {};
   \node (3) at (-0.1,1) [nd] { };
   \node[fill=gray] (4) at (1.1,1) [nd] {};
   \node (5) at (0.5,0) [nd] { };
   \draw[myptr] (5) to (3);
   \draw[myptr] (5) to (4);
   \draw[myptr] (3) to (1);
\node [inner sep = 0,below] at +(0.5,-0.5) {{\Large Directed 3-path}};
}

\def\DirTailedTri{
   \node (1) at (0.1,2) [nd] {$i$};
   \node (3) at (-0.1,1) [nd] { };
   \node[fill=gray] (4) at (1.1,1) [nd] {$j$};
   \node (5) at (0.5,0) [nd] { };
   \draw[myptr] (5) to (3);
   \draw[myptr] (5) to (4);
   \draw[myptr] (4) to (3);
   \draw[myptr] (3) to (1);
\node [inner sep = 0,below] at +(0.5,-0.5) {{\Large Directed tailed-triangle}};
}

\def\DirForClique{
  \node (1) at (0.1,2) [nd] {};
   \node (2) at (0.9,2) [nd] {$\ell$};
   \node (3) at (-0.1,1) [nd ] {$k$};
   \node (4) at (1.1,1) [nd] {$j$};
   \node (5) at (0.5,0) [nd] {$i$};
   \draw[myptr] (3) to (1);
   \draw[myptr] (4) to (1);
   \draw[myptr] (4) to (1);
   \draw[myptr] (3) to (2);
   \draw[myptr] (2) to (4);
   \draw[myptr] (5) to (1);
   \draw[myptr] (5) to (2);
   \draw[myptr] (4) to (3);
   \draw[myptr] (5) to (3); 
   \draw[myptr] (5) to (4);
\node [inner sep = 0,below] at +(0.5,-0.5) {{\Large Directed bipyramid}};
}

\begin{document}

% Copyright
\setcopyright{none}

% DOI
\doi{TBD}

% ISBN
\isbn{TBD}

\title{ESCAPE: Efficiently Counting All 5-Vertex Subgraphs}

\numberofauthors{3}

\author{
\alignauthor
Ali Pinar\\
\affaddr{Sandia National Laboratories\thanks{Sandia National Laboratories is a multi-program laboratory managed and operated by Sandia Corporation, a wholly owned subsidiary of Lockheed Martin Corporation, for the U.S. Department of Energy's National Nuclear Security Administration under contract DE-AC04-94AL85000.
}}\\
\affaddr{Livermore, CA}\\
\email{apinar@sandia.gov}
\alignauthor
C. Seshadhri\\
\affaddr{University of California}\\
\affaddr{Santa Cruz, CA}\\
\email{sesh@ucsc.edu}
\alignauthor
V. Vishal \\
\affaddr{ONU Technology}\\
\affaddr{Cupertino, CA}\\
\email{vishal@onutechnology.com}
}

\maketitle

\begin{abstract}
Counting the frequency of small subgraphs is a fundamental technique in network analysis across various domains, most notably in bioinformatics and social networks.  The special case of triangle counting has received much attention.  Getting results for 4-vertex or 5-vertex patterns is highly challenging, and there are few practical results known that can scale to massive sizes. 
% Indeed, even a highly tuned enumeration code takes more than a day on a graph with millions of edges. 
 
 We  introduce an algorithmic framework that can be adopted to count any small pattern in a graph and  apply this framework to  compute exact counts for  \emph{all} 5-vertex  subgraphs.  Our framework is built on cutting a pattern into smaller ones, and using counts of smaller patterns to get larger counts. Furthermore, we exploit degree orientations of the graph to reduce runtimes even further.  These methods avoid the combinatorial explosion that typical subgraph counting algorithms face.  We prove that it suffices to enumerate only four specific subgraphs (three of them have less than 5 vertices) to exactly count all 5-vertex patterns.

 We perform extensive empirical experiments on a variety of real-world graphs. We are able to compute counts of graphs with tens of millions of edges in  minutes on a commodity machine. To the best of our knowledge, this is the first practical algorithm for $5$-vertex pattern counting that runs at this scale.  A stepping stone to our main algorithm is a fast method for counting all $4$-vertex
patterns. This algorithm is typically ten times faster than the state of the art $4$-vertex counters. 
\end{abstract}


\section{Introduction}
\label{sec:introduction}

Subgraph counting is a fundamental network analysis technique used across diverse domains: bioinformatics, social sciences, and 
infrastructure networks studies~\cite{HoLe70,Co88,Po98,Milo2002,Burt04,PrzuljCJ04,Fa07,HoBe+07,BeBoCaGi08,Fa10,SzTh10,FoDeCo10,SonKanKim12}. 
The high frequencies of certain subgraphs in real networks gives a quantifiable method of proving
they are not Erd\H{o}s-R\'{e}nyi~\cite{HoLe70,WaSt98,Milo2002}. Distributions of small subgraphs are used
to evaluate network models, to summarize real networks, and classify vertex roles, among other things~\cite{HoLe70,PrzuljCJ04,Burt04,HoBe+07,Fa07,BeBoCaGi08,SeKoPi11,DuPiKo12,UganderBK13}.

The main challenge of motif counting is combinatorial explosion. 
As we see in our experiments, the
counts of $5$-vertex patterns are in the orders of billions to trillions,
even for graphs with a few million edges.
An enumeration algorithm is forced to touch each occurrence, and cannot terminate
in a reasonable time.
The key insight of this paper is to design a formal
framework of \emph{counting without enumeration} (or more precisely,
counting with minimal enumeration).
% 
% 
% 
% Even in a moderately sized graph with
% millions of edges, the subgraph counts (even for 4-vertex patterns) is in the billions.
% Any exhaustive enumeration method (no matter how cleverly designed) is forced to touch each occurrence
% of the subgraph, and cannot truly scale. One may apply massive parallelism to counteract this problem, but that does
% not avoid the fundamental combinatorial explosion. An alternative approach is based on \emph{sampling}. Here,
% we try to count the number of subgraphs using a randomized algorithm. The difficulty is in designing
% a fast algorithm that also provides an accurate estimate. The holy grail is to get
% mathematically provable bounds on accuracy with quantifiable error bars.
% 
Most existing methods~\cite{HoBe+07,BetzlerBFKN11,WongBQH12,RaBhHa14}
work for graphs of at most 100K edges, limiting their uses to
(what we would now consider) fairly small graphs.
A notable exception is recent work by Ahmed et al on counting 4-vertex patterns, that scales
to hundreds of millions of edges~\cite{AhNe+15}.

\subsection{The problem}

% \begin{figure}[t]
% \centering
% \resizebox{3in}{1.7in}{
%   \begin{tikzpicture}[nd/.style={circle,draw,fill=teal!50,inner sep=2pt},framed]    
%     \matrix[column sep=0.4cm, row sep=0.2cm,ampersand replacement=\&]
%     {
%      \ThreeStar \&    
%      \ThreePath \&  
%      \TailedTriangle  \\
%      \FourCycle  \& 
%      \ChordalFourCycle  \&
%      \FourClique     \\
%  };
%   \end{tikzpicture}  
%   }
% \caption{Connected 4-vertex patterns}
% \label{fig:4vertex}
% \end{figure}
% 
Our aim is simple: to exactly count the number of all vertex subgraphs (aka patterns, motifs, and graphlets)  up to size $5$ on massive graphs. There are 21 such connected subgraphs, as shown in~\Fig{5vertex}. Additionally, there are 11 disconnected patterns, which we discuss in \Sec{disc} of the Appendix. 
Throughout the paper, we refer to these subgraphs/motifs by their number. (Our algorithm also counts all $3$ and $4$ vertex patterns.)
We give a formal description in \Sec{problem}.

Motif-counting is an extremely popular research topic, and has
led to wide variety of results in the past years. 
As we shall see,  numerous approximate algorithms
that have been proposed for this problem~\cite{WernickeRasche06,TsKaMiFa09,BhRaRa+12,SePiKo13,RaBhHa14,JhSePi15}. Especially for validation at scale, it is critical to have a scalable, exact algorithm.
ESCAPE directly addresses this issue.

\begin{figure*}[t]
\centering
\resizebox{6in}{3.6in}{
  \begin{tikzpicture}[scale=0.8, every node/.style={scale=0.8}, nd/.style={circle,draw,fill=teal!50,inner sep=2pt},framed]    
    \matrix[column sep=0.4cm, row sep=0.2cm,ampersand replacement=\&]
    {
     \ThreeStar \&    
     \ThreePath \&  
     \TailedTriangle  \&
     \FourCycle  \& 
     \ChordalFourCycle  \&
     \FourClique  \\ \hline \\ 
     \pone \&    
     \ptwo \&  
     \pthree \& 
     \pfour \& 
     \pfive \&
     \psix \&
     \pseven    \\
      \peight \&    \pnine \&       \pten \&      \peleven \&      \ptwelve \&     \pthirteen \&     \pfourteen    \\
      \pfifteen \&    
     \psixteen \&  
     \pseventeen\& 
     \peighteen \& 
     \pnineteen \&
     \ptwenty \&
     \ptwentyone   \\};
  \end{tikzpicture}
  }  
\caption{Connected 4 and 5-vertex patterns}
\label{fig:5vertex}
\end{figure*}
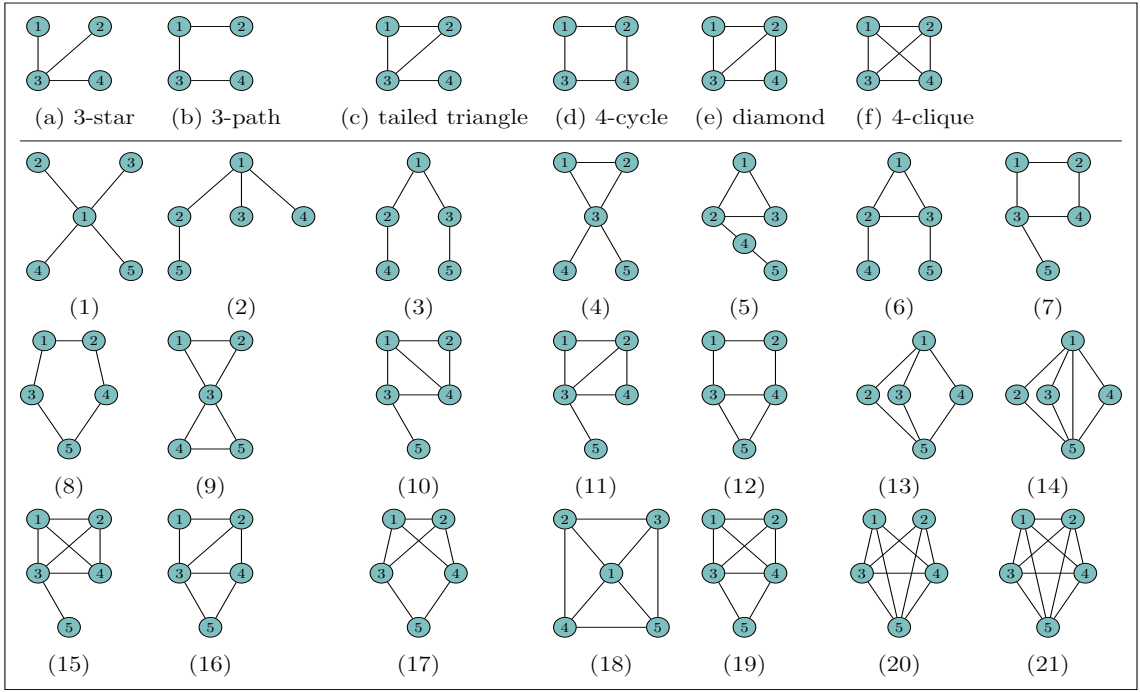

% \begin{figure}[t]
% \centering
%   \begin{tikzpicture}[nd/.style={circle,draw,fill=teal!50,inner sep=2pt},framed]    
%     \matrix[column sep=0.4cm, row sep=0.2cm,ampersand replacement=\&]
%     {
%       \drawfourstar \&
%       \drawfork \&
%       \drawpath \\
%     };
%   \end{tikzpicture}  
% \caption{5-vertex patterns}
%   \begin{tikzpicture}[nd/.style={circle,draw,fill=teal!50,inner sep=2pt},framed]    
%     \matrix[column sep=0.4cm, row sep=0.2cm,ampersand replacement=\&]
%     {
%       \directedtriangle{\linkab}{\linkba}{\linkba}{(a)}{} \&
%       \directedtriangle{\linkab}{\linkba}{\linkab}{(b)}{} \&
%       \directedtriangle{\linkaba}{\linkba}{\linkba}{(c)}{} \\
%       \directedtriangle{\linkaba}{\linkab}{\linkba}{(d)}{} \&
%       \directedtriangle{\linkaba}{\linkab}{\linkab}{(e)}{} \&
%       \directedtriangle{\linkaba}{\linkaba}{\linkba}{(f)}{} \\
%       \& \directedtriangle{\linkaba}{\linkaba}{\linkaba}{(g)}{} \&\\
%     };
%   \end{tikzpicture}
%   \caption{Directed triangles}
%     \label{fig:dtri}
% \end{figure}
% 
\subsection{Summary of our contributions}

We design the Efficient Subgraph Counting Algorithmic PackagE (ESCAPE), that produces
exact counts of all $\leq 5$-vertex subgraphs. We provide a detailed theoretical analysis
and run  experiments on a large variety of datasets, including 
web networks, autonomous systems networks, and social networks. All experiments are
done on a single commodity machine using 64GB memory.

\begin{asparaitem}%[\IEEEsetlabelwidth{This must be good}]
\item[\textbf{Scalability through careful algorithmics.}] Conventional wisdom is that $5$-vertex pattern
counting is not feasible because of size. There are a host of approximate methods, such as
color coding~\cite{HoBe+07,BetzlerBFKN11,ZhWaBu+12}, MCMC based sampling
algorithms~\cite{BhRaRa+12}, edge sampling algorithms~\cite{RaBhHa14, WernickeRasche06, Wernicke06}.
We challenge that belief. ESCAPE can do exact counting for  patterns up to 5 vertices on graphs with tens of millions of edges in a matter of minutes. 
(As shown in the experimental section of the above results, they do not scale
graphs of such sizes.) For instance, ESCAPE computes all 5 vertex counts on an router graph with 22M edges in under 5 minutes. 
%\Sesh{say something about real results}

% A point worth emphasizing is that our algorithm is a simple, sequential algorithm that
% runs on commodity hardware. We believe it is \emph{essential} to understand the limits
% of such classical algorithmic paradigms,  

% \item[\textbf{Fast and scalable.}] Our algorithm can  proddvide counts  for all connected patterns with up to 5-vertices, rapidly even for massive networks. 
% For instance, processing the {\tt as-skitter} graph with 1.7M vertices and 11.1M edges took less than 23 minutes. \Sesh{Say more.}
% %This is a task that was once deemed as ``practically intractable" beyond  hundreds of vertices. 

\item[\textbf{Avoiding enumeration by clever counting.}] One of the key insights into ESCAPE 
is that it suffices to enumerate a very small set of patterns to compute all 5 vertex counts.
Essentially, we build a formal framework of ``cutting" a pattern
into smaller subpatterns, and show that it is practically viable. From this theoretical framework,
we can show that it suffices
 to exhaustively enumerate a special (small) subset of patterns
to actually count all $5$-vertex patterns. 
Counting ideas to avoid enumeration have appeared
in the past practical algorithms~\cite{orca13,AhNe+15,ElShBo15,ElShBo16} but in a more ad hoc 
manner (and never for 5-vertex patterns.) 

The framework is absolutely critical for exact counting, since
enumeration is clearly infeasible even for graphs with a million edges. For instance, 
an automonous systems graph with 11M edges has $10^{11}$ instances of pattern 17 (in \Fig{5vertex}).
We achieve exact counts with clever data structures and combinatorial counting arguments. 

Furthermore, using standard inclusion-exclusion arguments, we prove that the counts of all
connected patterns can be used to get the counts of all (possibly disconnected) patterns.
This is done without any extra work on the input graph.  

\item[\textbf{Exploiting orientations.}]  A critical idea developed in ESCAPE
is  orienting edges in a degeneracy style ordering.
 Such techniques have been successfully applied to triangle counting before~\cite{ChNi85,Co09,SuVa11}. Here we show how this technique can be  extended to general pattern counting. 
This is what allows ESCAPE to be feasible for 5-vertex pattern counting,
and makes it much faster for 4-vertex pattern counting.

\item[\textbf{Improvements for 4-vertex patterns counting.}] The recent PGD package of Ahmed et al.~\cite{AhNe+15} has advanced the state of  art significantly with better 4-vertex pattern algorithms. ESCAPE is significantly faster, even by a factor of thousands in some instances. 

\item[\textbf{Trends in 5-vertex pattern counts:}]  Our ability to count 5-vertex patterns provide a powerful graph mining tool. While a thorough  analysis of 5-vertex counts is beyond the scope of this paper, we show a few examples on how our results can be potentially used for  edge prediction and graph classification.

\end{asparaitem}

% \textbf{Ability to count without enumeration:} A perceived bottleneck  for pattern counting has been the size of the output, which is a lower bound on runtime for enumeration based algorithms. Our 
% cutting framework overcomes this bottleneck by showing how to count larger
% patterns by piecing together counts for smaller patterns. Counting ideas to avoid
% enumeration have appeared
% in the past practical algorithms~\cite{orca13,AhNe+15,ElShBo15,ElShBo16} but in a more ad hoc manner (and never for 5-vertex
% patterns.) 
% 
% \textbf{Comprehensive results for 5-vertex patterns.} We get detailed results
% for all 4-vertex counts on a large number of graphs, and believe this is useful for further data analysis
% (as done in~\cite{UganderBK13}). Previous work usually focuses on a small, specific 
% set of larger motifs~\cite{HoBe+07,BetzlerBFKN11,ZhWaBu+12}, or gives coarser approximations
% for more motifs~\cite{BhRaRa+12}.

\begin{figure}[th]
\centering
 \includegraphics[width=0.9\columnwidth]{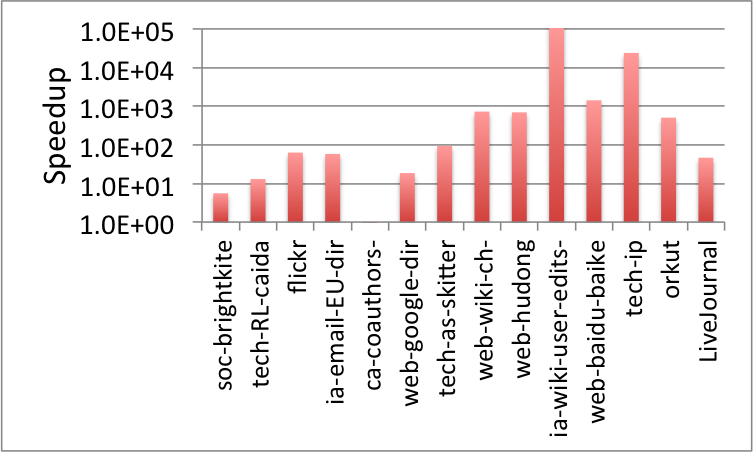}
 \caption{\label{fig:4speedup} \small Speedup achieved by Escape over PGD (computed as runtime of PGD/runtime of ESCAPE). } 
  \end{figure}

\subsection{Related Work} \label{sec:related}

In fields as varied as social sciences, biology, and physics,
it has been observed that the frequency of small pattern subgraphs
plays an important role in graph structure~\cite{HoLe70,Co88,WaFa94,WaSt98,EcMo02,Milo2002,Burt04,Fa07,BeBoCaGi08,FoDeCo10,Fa10,SzLaTh10,SzTh10,SonKanKim12}. 
Specifically in bioinformatics, pattern counts have significant
relevance in graph classiciation~\cite{Milo2002,PrzuljCJ04,HalesA08}.

In social networks, Ugander et al.~\cite{UganderBK13}, underlined the significance of 4-vertex patterns by  proposing a ``coordinate system'' for graphs based on the motifs distribution. This was applied to classification  of  comparatively small  networks (thousands of vertices).
We stress that this was useful even without graph attributes, and thus
the structure itself was enough for classification purposes.
A number of recent results have used small subgraph counts for detecting
communities and dense subgraphs~\cite{SaSe15,Ts15,BeGlLe16,TsPaMi16}.

From the practical algorithmics standpoint, triangle counting
has received much attention. We simply refer the reader
to the related work sections of~\cite{TsKoMi11,SePiKo13}.
Gonen and Shavitt~\cite{GonenS09} propose exact and approximate algorithms for computing {\em non-induced} counts of some 4-vertex motifs. 
They also consider counting number of motifs that a vertex participates in, an instance of a problem called {\em motif degree counting}, which has gained a lot of attention recently (see \cite{MilenkovicP08, StoicaP09, GonenS09, birmele2012}).
Marcus and Shavitt~\cite{MarcusS10} give exact algorithms for computing all 4-vertex motifs running in time $O(d \cdot m + m^2)$. Here $d$ is the maximum vertex degree  and $m$ is the number of edges. 
Their package RAGE does not scale to large graphs. The largest graph processed has 90K edges and takes 40 minutes. They compare with the bioinformatics FANMOD package~\cite{WernickeRasche06, Wernicke06}, which takes about 3 hours.

A breakthrough in exact 4-vertex pattern counting was recently achieved by Ahmed et al.~\cite{AhNe+15,AhNe-pre15}. Using techniques on graph transitions based on edge addition/removal,
their PGD (Parametrized Graphlet Decomposition) package handles graphs with tens of millions of edges and more, and is many orders of magnitude faster
than RAGE. It routinely processes 10 million edge graphs in under an hour.
There are other results on counting 4-vertex patterns, but none achieve the scalability of PGD~\cite{HoDe16,orca13}.
We consider PGD to be the state-of-the-art for 4-vertex pattern counting. They
do detailed comparisons and clearly outperform previous work. (Notably, the authors made
their code public~\cite{PGD}.)

% ESCAPE is significantly faster than PGD, as we show in our experimental results.

% Example for a graph on 90K edges they state that their enumeration algorithm takes 40 minutes. They compare with the previous best enumeration algorithm of FANMOD~\cite{WernickeRasche06, Wernicke06}. On the same graph FANMOD takes 2 hour 48 minutes. 
% %We note that the enumeration algorithm that we compare our results with are much much faster than \cite{MarcusS10}. 

Elenberg et al.~\cite{ElShBo15,ElShBo16} give algorithms for computing pattern \emph{profiles},
which involve computing pattern counts per vertex and edge. This is a significantly harder
problem, and Elenberg et al. employ approximate and distributed algorithms. The maximum
graph size they handle is in order of tens of millions of edges.

Many of the results above~\cite{AhNe+15,HoDe16,ElShBo16} use combinatorial strategies 
to cut down enumeration, which our cutting framework tries to formalize. For the special
case of vertex and edge profiles, Melckenbeeck et al. give an automated method
to generate combinatorial equations for profile counting~\cite{MAM+16}. These results
only generate linear equations, and do not prescribe the most efficient method of counting.
In contrast, our cutting framework generates polynomial formulas, and we deduce
the most efficient formula for 5-vertex patterns.

As an alternative exact counting, Jha et al.~\cite{JhSePi15} proposed 3-path sampling to estimate all 4-vertex counts.  Their technique builds on \emph{wedge sampling}~\cite{ScWa05-2,SePiKo13,KoPiPlSe13} and samples paths of length 3 to estimate
various 4-vertex statistics. 
% Their techniques  come with  error bars  for the estimation, and can provide accurate estimates with in seconds even for the  graphs with hundreds of millions of edges.  

To the best of our knowledge, there is no method (approximate or exact) that can count all 5-vertex
patterns for graphs with millions of edges.

\section{Preliminaries} 
\subsection{Formal description of the problem}\label{sec:problem}

Our input is an undirected simple graph $G=(V,E)$, with $n$ vertices and $m$ edges.
We distinguish subgraphs from \emph{induced subgraphs}~\cite{Diestel06}.
A subgraph is a subset of edges. An induced subgraph is obtained
by taking a subset $V'$ of vertices, and taking \emph{all edges}
among these vertices. 

\emph{We wish to get induced and non-induced counts for all patterns up to size $5$.}
As shown later in \Thm{disc}, it suffices to get counts for only connected patterns,
since all other counts can be obtained by simple combinatorics.
The connected $4$-vertex and $5$-vertex patterns are shown in \Fig{5vertex}. 
%(We use the term \emph{$k$-pattern} to refer to a $k$-vertex pattern.)
%
For convenience, the \emph{$i$th $4$-pattern}
refers to $i$th subgraph with 4 vertices  in \Fig{5vertex}.
For example, the $6$th $4$-pattern in the 
four-clique and the $8$th $5$-pattern is the five-cycle.

Without loss of generality, we focus on computing induced subgraph counts. 
We use $\ind_i$ (resp. $\nonind_i$) to denote the induced (resp. non-induced)
count of the $i$th $5$-pattern. \emph{Our aim is to compute all $\ind_i$ values.}
A simple (invertible) linear transformation gives all the $\ind_i$ values
from the $\nonind_i$ values. 
we provide details in \Sec{transform} of the Appendix.

%This matrix is given in \Fig{matrix}, where the $(i,j)$the entry is the number of distinct copies of the $i$th $5$-pattern in the $j$th one.
% 
% 
% 
% It is technically convenient to think of induced subgraph counts. We 
% denote the number of induced occurrences of the $i$-th subgraph (of \Fig{4-vertex-motifs})
% by $\ind_i$. So, $\ind_4$ is the number of induced 4-cycles in $G$,
% which is the number of distinct subsets of 4 vertices that induce a 4-cycle.
% When we talk of a ``vanilla" subgraph, we mean the usual subgraph setting
% (a subset of edges). In general, if we do not say ``induced", we mean
% vanilla.
% 
% \emph{Our aim is to get an estimate of all $\ind_i$ values.} Let $\nonind_i$
% to denote the number of (vanilla) subgraph occurrences of the $i$th subgraph,
% There is a simple linear relationship between induced and non-induced counts,
% given below. The $(i,j)$ entry of the matrix $A$ below 
% is simply the number of distinct copies of the $i$th
% subgraph in the $j$th subgraph (so $A_{2,4} = 4$, the number of 3-paths
% in the 4-cycle).
%

\subsection{Notation} \label{sec:notation}

The input graph $G = (V,E)$ is undirected and has $n$ vertices and $m$ edges. For analysis, 
we assume that the graph is stored as an adjacency list, where each list is a hash table.
Thus, edge queries can be made in constant time.  

We denote the \emph{degree ordering} of $G$ by $\prec$. For vertices $i,j$, we say $i \prec j$,
if either $d(i) < d(j)$ or $d(i) = d(j)$ and $i < j$ (comparing vertex id). 
We construct the \emph{degree ordered DAG} $G^\rightarrow$ by orienting all edges in $G$ according to $\prec$.

Our results and proofs are somewhat heavy on notation, and important terms
are provided in \Tab{notation}. Counts of certain patterns, especially
those in  \Fig{5vertex-basis}, will receive
special notation. Note that some of these patterns are directed,
since we will require the count of them in $G^\rightarrow$.

We will also need \emph{per-vertex, per-edge}
counts for some patterns. For example, $T(G)$ denotes the total
number of triangles, while $T(i), T(e)$ denote the number of triangles
incident to vertex $i$ and edge $e$ respectively.

\begin{center}
\begin{tabular}{c || c}
Notation & Count \\ \hline \hline
$d(i)$ & degree of $i$ \\ \hline 
$W(G)$ & wedge \\ \hline
$W_{++}(G^\rightarrow), W_{+-}(G^\rightarrow)$ & out-wedge, inout-wedge \\ \hline
$W(i,j)$ & wedge between $i,j$ \\ \hline
$W_{++}(i,j)$ & outwedge between $i,j$ \\ \hline
$W_{+-}(i,j)$ & wedge from $i$ to $j$ \\ \hline
$T(G), T(i), T(e)$ & triangle \\ \hline
$C_4(G), C_4(i), C_4(e)$ & 4-cycle \\ \hline
$K_4(G), K_4(i), K_4(e)$ & 4-clique \\ \hline
$D(G)$ & diamond \\ \hline
$DD(G^\rightarrow)$ & directed diamond \\ \hline
$DP(G^\rightarrow)$ & directed 3-path \\ \hline
$DBP(G^\rightarrow)$ & directed bipyramid \\ \hline
\end{tabular} \label{tab:notation}
\end{center}

% We denote the degree of a vertex by $d(i)$. We denote the number of triangles by $T(G)$,
% and use $T(i)$, $T(e)$ to denote the number of triangles incident to vertex $i$
% and edge $e$ respectively. Analogously, we use $C_4(G)$, $C_4(i)$, and $C_4(e)$
% to denote the number of $4$-cycles (pattern (d) in~\Fig{4vertex}). Similarly,
% $K_4(\cdot)$ is used to denote $4$-clique counts. We use $\cc(\cdot)$ to denote
% the number of chordal-cycles (pattern (e) in~\Fig{4vertex}).  
% For all patterns in~\Fig{5vertex}, we use $N_i$ to denote the non-induced
% count of pattern $i$. 
% 
% A \emph{wedge} is a path of length $2$. The total number of wedges is denoted by $W$,
% and $W(i,j)$ is the number of wedges with vertices $i$ and $j$ as endpoints. In $D$,
% there are three types of wedges: out-wedges, inout-wedges, and in-wedges (refer to \Fig{4vertex-basis} and \Fig{5vertex-basis})
% The number of out-wedges is denoted by $W_{++}(D)$, and the number of out-wedges with
% endpoints $i$, $j$ is denoted as $W_{++}(i,j)$. Analogously, the number of inout-wedges
% and in-wedges are denoted by $W_{+-}(\cdot)$ and $W_{--}(\cdot)$. 

\section{Main theorems} \label{sec:mainthm}

Our final algorithms are quite complex and use a variety of combinatorial methods
for efficiency. Nonetheless, the final asymptotic runtimes are easy to express.
(While we do not focus on this, the leading constants in the $O(\cdot)$ are quite small.)
Our main insight is: despite the plethora of small subgraphs, it suffices to
enumerate a very small, carefully chosen set of subgraphs to count everything else.
Furthermore, these subgraphs can themselves be enumerated 
with minimal overhead.

% \begin{figure}[t]
% \centering
%   \begin{tikzpicture}[scale=0.6,nd/.style={scale=0.75,circle,draw,fill=teal!50,inner sep=2pt},minimum width = 13pt, framed]    
%     \matrix[column sep=0.4cm, row sep=0.2cm,ampersand replacement=\&]
%     {
%      \Outwedge  \& 
%      \Inoutwedge \&
%      \DirChordal \\
%  };
%   \end{tikzpicture}  
% \caption{\small Fundamental patterns for 4-vertex pattern counting}
% \label{fig:4vertex-basis}
% \end{figure}
% 
\begin{theorem} \label{thm:fourvertex}  
There is an algorithm for exactly counting all connected 4-vertex patterns in $G$ whose runtime is $O(W_{++}(G^\rightarrow) + W_{+-}(G^\rightarrow) + \dcc(G^\rightarrow) + m + n)$ 
and storage is $O(n+m)$.
\end{theorem}

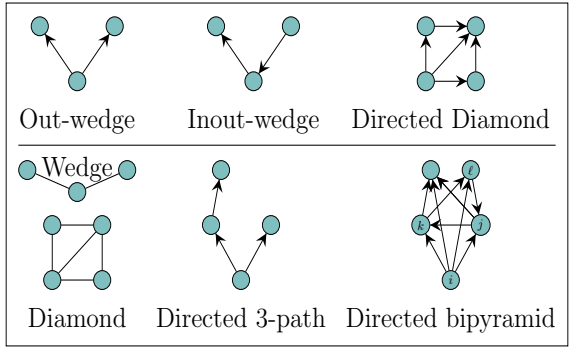
\begin{figure}[t]
\centering
\resizebox{3in}{1.8in}{
  \begin{tikzpicture}[nd/.style={scale=0.75,circle,draw,fill=teal!50,inner sep=2pt},minimum width = 13pt, framed]    
    \matrix[column sep=0.4cm, row sep=0.2cm,ampersand replacement=\&]
    {
     \Outwedge  \& 
     \Inoutwedge \&
     \DirChordal \\ \hline \\
     \Wedge \& 
     \DirThreePath \&
     \DirForClique \\
 };
  \end{tikzpicture}  
}
\caption{\small Fundamental patterns for 4-vertex (above) and 5-vertex (below) pattern counting}
\label{fig:5vertex-basis}
\end{figure}

\begin{theorem} \label{thm:fivevertex} There is an algorithm for exactly counting all connected 5-vertex patterns in $G$ 
whose runtime is $O(W(G) + \cc(G) + \dpath(G^\rightarrow) + \dbp(G^\rightarrow) + m + n)$ 
and storage complexity is $O(n+m+T(G))$.
\end{theorem}

Note that previous theorems only handle connected patterns. But a routine inclusion-exclusion
argument yields the following theorem.

\begin{theorem} \label{thm:disc} Fix a graph $G$. Suppose we have
counts for all connected $r$-vertex patterns, for all $r \leq k$.
Then, the counts for all (even disconnected) $k$-vertex patterns
can be determined in constant time (only a function of $k$). 
\end{theorem} 

\textbf{Outline of remaining paper:} \Sec{ideas} gives a high level overview
of our main techniques. \Sec{cut} discusses the cutting framework used to reduce
counting all patterns into enumeration of some specific patterns (namely, those
in \Fig{5vertex-basis}). In \Sec{fourvertex},
we apply this framework to 4-vertex pattern count, and prove \Thm{fourvertex}.
In \Sec{fivevertex}, we work towards 5-vertex pattern counting and prove \Thm{fivevertex}.

The proof and discussion of \Thm{disc} is omitted because of space constraints
and appears in \Sec{disc} of the Appendix. 
%our full paper~\cite{escape-full}. 
In the remainder of this main body of the paper, we only focus on connected patterns.

\section{Main ideas} \label{sec:ideas}

The goal of ESCAPE is to avoid the combinatorial explosion that occurs in a typical enumeration algorithm.
For example, the {\tt tech-as-skitter} graph has 11M edges, but \emph{2 trillion} 5-cycles.
Most 5-vertex pattern counts are upwards of many billions
for most of the graphs (which have only ~10M edges).
Any algorithm that explicitly touches each such pattern is bound to fail. The second difficulty is that 
the time for enumeration is significantly \emph{more} than the count of patterns. This is because we have to find
\emph{all potential} patterns, the number of which is more than the
count of patterns. A standard
method of counting triangles is to enumerate wedges, and check whether it participates in a triangle. 
The number of wedges in a graph is typically an order
of magnitude higher than the number of triangles.

\textbf{Idea 1: Cutting patterns into smaller patterns.} For a pattern $H$, a cut set 
is a subset of vertices whose removal disconnects $H$. Other than the clique,
every other pattern has a cut set that is a strict subset of the vertices (we call
this a non-trivial cut set).
Formally, suppose there is some set of $k$ vertices $S$, whose
removal splits $H$ into connected components $C_1, C_2, \ldots$. Let the graphs induced
by the union of $S$ and $C_i$ be $H_i$. The key observation is that if we determine
the following quantities, we can count the number of occurrences of $H$.
\begin{asparaitem}
    \item For each set $S$ of $k$ vertices in $G$, the number of occurrences of $H_1, H_2, \ldots$
that involve $S$.
    \item The number of occurrences of $H'$, for all $H'$ with fewer vertices than $H$.
\end{asparaitem}
The exact formalization of this requires a fair bit of notation and the language of graph
automorphisms. This gives a set of (polynomial) formulas for counting most
of the 5-vertex patterns. These formulas can be efficiently evaluated with appropriate data structures.

There is some art in choosing the right $S$ to design the most efficient algorithm.
In most of the applications, $S$ is often just a vertex or edge. Thus, if we know the number
of copies of $H_i$ incident to every vertex and edge of $G$, we can count $H$.
This information can be determined by enumerating all the $H_i$s, which 
is a much simpler problem. 

\textbf{Idea 2: Direction reduces search.} A classic algorithmic idea for triangle counting
is to convert the undirected $G$ into the DAG $G^\rightarrow$, 
and search for directed triangles~\cite{ChNi85,ScWa05,Co09}. We extend this approach
to 4 and 5-vertex patterns. The idea is to search for \emph{all} non-isomorphic DAG
versions of the pattern $H$ in $G^\rightarrow$. 
This is combined with Idea 1, where we break up  patterns in smaller ones.
These smaller patterns are enumerated through $G^\rightarrow$,
since the direction significantly cuts down the combinatorial expansion
of the enumeration procedure. The use of graph orientations has
been employed in theoretical algorithms for subgraph counting~\cite{AlYuZw94}.
We bring this powerful technique to practical counting of 4 and 5-vertex patterns.
% 
% When $H$ is not a clique, we can apply Idea 1 to design algorithms
% to count each of these directed version. The key insight is that the resulting
% enumeration problems are significantly faster to perform, since the direction
% cuts down the combinatorial explosion. 
% 
% When $H$ is a clique, we simply perform an enumeration that builds smaller directed cliques,
% and tries to piece them together to get larger ones.
% This gives major speedup over an undirected enumeration.
% This is well documented for counting triangles~\cite{ScWa05}, and we observe it again for counting
% $4$ and $5$ cliques. 

\section{The cutting framework} \label{sec:cut} 
This section introduces   the framework of our algorithms. We start with  introducing the theory, and then  discuss how it can  used for algorithm design and present its application to 5-pattern 2 counting.

Let $H$ be a pattern we wish to count in $G$.
For any set of vertices $C$ in $H$, $H|_C$ is the subgraph of $H$ induced on $C$.
For this section, it is convenient to consider $G$ and $H$ as labeled.
This makes the formal analysis much simpler. (Labeled counts can
be translated to unlabeled counts by pattern automorphism counts.)

We formally define a match and a partial match of the patterm $H = (V(H),E(H))$.
As defined, a match is basically an induced subgraph of $G$ that is exactly $H$.

\begin{definition} \label{def:copy} A \emph{match} of $H$ is a bijection $\pi: S \rightarrow V(H)$
where $S \subseteq V$ and $\forall s_1, s_2 \in S$, $(s_1, s_2)$ is an edge of $G$
iff $(\pi(s_1), \pi(s_2))$ is an edge of $H$.
The set of distinct matches of $H$ in $G$ is denoted $\match(H)$. 

If $\pi$ is only an injection (so $|S| < |V(H)|$), then $\pi$ is a \emph{partial match}.
% The set of distinct partial matches is denoted $\pmatch(H)$.

A match $\pi:S \rightarrow V(H)$ \emph{extends} a partial match $\sigma: T \rightarrow V(H)$
if $S \supset T$ and $\forall t \in T$, $\pi(t) = \sigma(t)$.
\end{definition} 

\begin{definition} \label{def:deg} Let $\sigma$ be a partial match of $H$ in $G$.
The $H$-degree of $\sigma$, denoted $\deg_H(\sigma)$,
 is the number of matches of $H$ that extend $\sigma$.
\end{definition}

We now define the \emph{fragment} of $G$ that is obtained by cutting $H$ into smaller patterns.

\begin{definition} \label{def:frag} Consider $H$ with some non-trivial cut set $C$ (so $|C| < |V(H)|$),
whose removal leads to connected components $S_1, S_2, \ldots$. The \emph{$C$-fragments} of $H$
are the subgraphs of $H$ induced by $C \cup S_1, C \cup S_2, \ldots$. This set is denoted by $\frag_C(H)$.
\end{definition}

Before launching into the next definition, it helps to explain the main cutting lemma.
Suppose we find a copy $\sigma$ of $H|_C$ in $G$. If $\sigma$ extends to a copy of every possible $F_i \in \frag_C(H)$
\emph{and} all these copies are disjoint, then they all combine to give a copy of $H$.
When these copies are not disjoint, we end up with another graph $H'$, which we call a \emph{shrinkage}.

\begin{definition} \label{def:shrinkage} Consider graphs $H$, $H'$, and a non-trivial cut set $C$
for $H$. Let the graphs in $\frag_C(H)$ be denoted by $F_1, F_2, \ldots$.
A \emph{$C$-shrinkage of $H$ into $H'$} is a set of maps $\{\sigma, \pi_1,$ $\pi_2, \ldots,$ $\pi_{|\frag_C(H)|}\}$
with the following properties.
\begin{asparaitem}
    \item $\sigma: H|_C \rightarrow H'$ is a partial match of $H'$.
    \item Each $\pi_i: F_i \rightarrow H'$ is a partial match of $H'$.
    \item Each $\pi_i$ extends $\sigma$.
    \item For each edge $(i,j)$ of $H'$, there are some index $c \in |\frag_C(H)|$
    and vertices $a,b \in F_i$ such that $\pi_i(a) = i$ and $\pi_i(b) = j$.
\end{asparaitem}

The set of graphs $H'$ such that there exists some $C$-shrinkage of $H$ in $H'$ is denoted
$\shrink_C(H)$. For $H' \in \shrink_C(H)$, the number of distinct $C$-shrinkages
is $\numshrink_C(H,H')$.
\end{definition}

The main lemma tells us that if we know $\deg_F(\sigma)$ for every copy $\sigma$
of $H|_C$ and for every $C$-fragment $F$, and we know the counts of every possible
shrinkage, we can deduce the count of $H$. 

\begin{lemma} \label{lem:cut} Consider pattern $H$ with cut set $C$. Then,
\begin{eqnarray*}
 \match(H) & = & \sum_{\sigma \in \match(H|_C)} \prod_{F \in \frag_C(H)} \deg_F(\sigma) \\
& & - \sum_{H' \in \shrink_C(H)} \numshrink_C(H,H') \match(H') 
\end{eqnarray*}
\end{lemma}

\begin{proof} Consider any copy $\sigma$ of $H|_C$. Take all tuples of the form $(\pi_1, \pi_2, \ldots, \pi_{|\frag_C(H)|})$
where $\pi_i$ is a copy $F_i \in \frag_C(H)$ that extends $\sigma$. The number of such tuples
is exactly $\sum_{\sigma \in \match(H|_C)} \prod_{F \in \frag_C(H)} \deg_F(\sigma)$. 

Abusing notation, let $V(\pi_i)$ be the set of vertices that $\pi_i$ maps to $F_i$. If all $V(\pi_i) \setminus V(C)$ are distinct,
by definition, we get a copy of $H$. If there is any intersection, this is a $C$-shrinkage of $H$ into 
some $H'$. Consider aggregating the above argument over all copies $\sigma$. Each match of $H$ is counted
exactly once.  Each match of $H' \in \shrink_C(H)$ is counted for every distinct $C$-shrinkage
of $H$ into $H'$, which is exactly $\numshrink_C(H,H')$. This completes the proof.
\end{proof}

\textbf{Algorithmically using this lemma:} Suppose $H$ is a $5$-vertex
pattern, and counts for all $\leq 4$-vertex patterns are known. In typical examples, $C$
is either a vertex or an edge. Thus, each $\sigma$ in the formula is
simply just every possible vertex or edge. If we can enumerate
all matches of each $F \in \frag_C(H)$, then we can store $\deg_F(\sigma)$
in appropriate data structures. Each $F$ has strictly less than $5$-vertices
(and in most cases, just $2$ or $3$), and thus, we can hope to enumerate $F$.

Once all $\deg_F(\sigma)$ are computed, we can iterate over all $\sigma$
to compute the first term in \Lem{cut}. We need to subtract out the summation
over $\shrink_C(H)$. Observe that $\numshrink_C(H,H')$ is an
absolute constant independent of $G$, so it can be precomputed. Each $H' \in \shrink_C(H)$
has less than $5$ vertices, so we already know $\match(H')$.

This yields $\match(H)$. To get the final \emph{unlabeled} frequency,
we must normalize to $\match(H)/|\aut(H)|$.
(Here, $\aut(H)$ is the set of automorphisms of $H$. The same unlabeled
pattern can be counted multiple times as a labeled match. For example,
every triangle gets counted three times in $\match$, and we divide
this out to get the final unlabeled frequency.)

%  Note that all $F$ and shrinkages have less vertices
% than $H$, and thus we can build upon counts of smaller vertex patterns to count $H$.
% 
% \begin{definition} \label{def:aut} The set of distinct matches of $H$ in $H$ (these are ``self-matches") is the automorphism set of $H$, denoted
% by $\aut(H)$.
% % Consider a set of vertices $T \in V(H)$ represented as an ordered list $[t_1, t_2, \ldots]$.
% % Two ordered lists of vertices $T = [t_1, t_2, \ldots]$ and $T' = [t'_1, t'_2, \ldots]$ are \emph{related}
% % if there exists $\pi \in \aut(H)$ such that $\forall i$, $\pi(t_i) = t'_i$. (This forms an equivalence relation.) 
% % The set of lists that are related to $T$ is the \emph{orbit} of $T$ in $H$, denote $\orb_T(H)$.
% \end{definition}

\textbf{Application of lemma for pattern 2:} To demonstrate this lemma, 
let us derive counts for $5$-pattern (2). We use the labeling in \Fig{5vertex}.
 Let edge $(1,2)$ be the cut set $S$. Thus, the fragments are $F_1$, the wedge $\{(1,2), (2,5)\}$
and $F_2$, the three-star $\{(1,2), (1,3), (1,4)\}$. 
Every edge in $G$ is a match of $H|_S$. Consider $(i,j)$ with match
$\sigma(i) = 1$ and $\sigma(j) = 2$. The degree $\deg_{F_1}(\sigma)$ is $d(j) - 1$. 
The degree $\deg_{F_2}(\sigma)$ is $(d(i)-1)(d(j)-2)$. 

The only possible shrinkage of the patterns is into a tailed triangle. Let $H$
be the $5$-pattern (2), and $H'$ be the tailed triangle.
Note that $\numshrink_C(H,H')$ is $2$. 
In both cases, we set $\sigma'(1) = 3$ and $\sigma'(2) = 1$. Set $\pi_1(5) = 2$ and $\pi_2(3) = 2$, $\pi_2(4) = 4$.
Alternately, we can change $\pi_2(4) = 2$ and $\pi_2(3) = 4$. The set of maps $\{\sigma', \pi_1, \pi_2\}$ in both cases
forms a $C$-shrinkage of this pattern into tailed triangles.

This
\begin{eqnarray*}
 & & \match(H) = \sum_{(i,j) \in E} \big[ (d(j)-1)(d(i)-1)(d(i)-2) \\
 & & + (d(i)-1)(d(j)-1)(d(j)-2)\big] - 2\cdot\match(\textrm{tailed triangle}) 
\end{eqnarray*}

Note that $H$ has two automorphisms, as does the tailed triangle. Thus, the number of tailed
triangle matches (as a labeled graph) is twice the frequency. A simple argument shows that
the number of tailed triangles is $\sum_i t(i)(d(i)-2)$ (we can also derive this from \Lem{cut}).
Thus, 
$$ N_2 = \sum_{\langle i,j\rangle \in E} (d(j)-1){d(i)-1 \choose 2} - 2\sum_i t(i)(d(i)-2) $$

\section{Counting 4-vertex patterns} \label{sec:fourvertex}

A good introduction to these techniques is counting $4$-vertex patterns.
The following formulas have been proven in~\cite{JhSePi15,AhNe+15}, but can be derived using the framework of \Lem{cut}.

\begin{theorem} \label{thm:4vertex-easy} 
$\textrm{\# 3-stars}  =  \sum_i {d(i) \choose 3}$, 
$\textrm{\# diamonds}  =  \sum_e {t(e) \choose 2}$,

$\textrm{\# 3-paths}  =  \sum_{(i,j) \in E} (d(i)-1)(d(j)-1) - 3\cdot T(G)$,

$\textrm{\# tailed-triangles}  =  \sum_i t(i)(d(i)-2)$ 
\end{theorem}

For counting $4$-cycles, note any set of opposite
vertices (like $1$ and $4$ in 4-cycle of~\Fig{5vertex}) form a cut. It is easy to see that
$C_4(G) = \sum_{i < j} {W(i,j) \choose 2}/2$. These values are potentially expensive
to compute as a complete wedge enumeration is required. Employing the degree ordering,
we can prove a significant improvement. 
With a little care, we can get counts per edge. (We remind the reader
that $\succ$ refers to the degree ordering.)

\begin{theorem} \label{thm:4cycle} $ C_4(G) = \sum_{i \succ j} {W_{++}(i,j) + W_{+-}(i,j) \choose 2}$.
For edge $(i,k)$ where $i \succ k$, 

$C_4((i,k)) = \sum_{j \succ k}$ $[W_{++}(i,j) +$ $W_{+-}(i,j)$
$+W_{+-}(j,i)-1] + $ $\sum_{j \prec k} [W_{+-}(j,i) + W_{++}(i,j)-1]$.
\end{theorem}

\begin{proof} Consider all DAG versions of the 4-cycle, as given in \Fig{4cycle-DAG}.
Let $i$ denote the highest vertex according to $\prec$, and let $j$ be the opposite
end (as shown in the figure). The key observation is that wedges between $i$ and $j$
are either 2 out-wedges, 2 inout wedges, or one of each. Summing over all possible $j$s,
we complete the proof of the basic count.

Consider edge $(i,k)$ where $i \succ k$. To determine the 4-cycles on this edge, we look
at all wedges that involve $(i,k)$. Suppose the third vertex on such a wedge is $j$.
We have two possibilities. (i) $i \succ k \prec j$: Thus, $(i,k,j)$ is an out wedge,
and could be a part of a 4-cycle of type (a) or (c). Any out or inout wedge 
between $i$ and $j$ creates a 4-cycle. We need a $-1$ term to subtract out the
wedge $(i,k,j)$ itself.

(ii) $i \succ k \succ j$: This is an inout wedge and can be part of a 4-cycle
of type (b) or (c). Again, any other out or inout wedge between $i$ and $j$
forms a 4-cycle. A similar argument to the above completes the proof.
% 
% 
% at the wedges using this wedge. Consider some wedge involving $j$. When we take $j \succ k$, we could be in type (a)
% or type (c). If $j \prec k$, we are considering type (b). Summing over the possibilities of the other wedge,
% we complete the proof.
\end{proof}

% For our next theorem, we use a classic result to enumerate triangles.
% 
% \begin{theorem} \label{thm:tri} \cite{ChNi85,ScWa05} Triangles can be enumerated in $O(W_{++}(G) + n + m)$ time
% and $O(m+n)$ space. (Thus, all values $T(i)$ and $T(e)$ can computed in this time.)
% \end{theorem}
% 
Now, we show how to count $4$-cliques.

\begin{theorem} \label{thm:4clique} (We remind the reader that $\dcc$ is the number
of directed diamonds, as shown in \Fig{5vertex-basis}.) The number of four-cliques per-vertex and per-edge 
can be found in time $O(W_{++}(G^\rightarrow) + \dcc(G^\rightarrow))$ 
and $O(m)$ additional space.
\end{theorem}

\begin{proof} Let $H$ denote the directed diamond of \Fig{5vertex-basis}.
The key observation is that every four-clique in the original graph
must contain one (and exactly one) copy of $H$ as a subgraph.
% is the smallest vertex according to $\prec$. 
It is possible to enumerate
all such patterns in time linear in $\dcc(G^\rightarrow)$. 
We simply loop over all edges $(i,j)$,
where $i \prec j$. We enumerate all the outout wedges involving $(i,j)$,
and determine all triangles involving $(i,j)$ with $i$ as the smallest vertex.
Every pair of such triangles creates a copy of $H$, where $(i,j)$ forms
the diagonal. For each such copy,
we check for the missing edge to see if it forms a four-clique.
Since we enumerate all four-cliques, it is routine to find the per-vertex and per-edge
counts.
\end{proof}

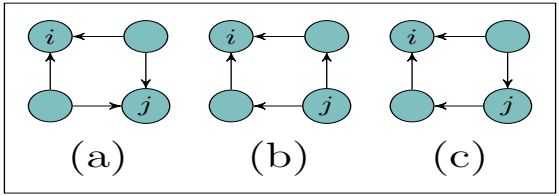
\begin{figure}[t]
\centering
\resizebox{3in}{1in}{
  \begin{tikzpicture}[nd/.style={circle,draw,fill=teal!50,inner sep=2pt},minimum width = 13pt, framed]    
    \matrix[column sep=0.4cm, row sep=0.2cm,ampersand replacement=\&]
    {
     \FourCycleOne  \& 
     \FourCycleTwo \&
     \FourCycleThree \\
%     \DirChordal \\
 };
  \end{tikzpicture}  
  }
\caption{All acyclic orientations of the 4-cycle}
\label{fig:4cycle-DAG}
\end{figure}

We state below a stronger version of the 4-vertex counting theorem, \Thm{fourvertex}.
This will be useful for 5-vertex pattern counting.
% 
% of the ability to enumerate triangles. This will be used for 5-vertex counting.

\begin{theorem} \label{thm:4vertex-tri} In $O(W_{++} + W_{+-} + \dcc + m)$ time and $O(T)$ additional space,
there is an algorithm that computes (for all vertices $i$, edges $e$, triangle $t$): all $T(i)$, $T(e)$, $C_4(i)$, $C_4(e)$, $K_4(i)$, $K_4(e)$, $K_4(t)$ counts,
and for every edge $e$, the list of triangles incident to $e$.
\end{theorem}

\begin{proof} A classic theorem basically states that all triangles
can be enumerated in $O(W_{++}(G^\rightarrow))$ time~\cite{ChNi85,ScWa05}. (We used the
same argument to prove \Thm{4clique}.) By \Thm{4vertex-easy}, once 
we have per-vertex and per-edge triangle counts, we can count everything
other that $4$-cycles and $4$-cliques in linear time. By \Thm{4cycle}, enumerating
outout and inout wedges suffices for 4-cycle counting. We add the bound
on \Thm{4clique} to complete the proof.
\end{proof}

% 
% \begin{proof} Let us first compute just the total number
% of four-cycles. The key observation is given in \Fig{4cycle-DAG}, which lists all possible non-isomorphic orientations of the four-cycle. Note that in $D$, each four-cycle appears in one of those three
% forms. Observe crucially that the wedge with $i$ and $j$ as ends is either an outout or inout 
% wedge.
% Thus, we can express the number of 4-cycles as
% $\sum_i \sum_{j \prec i} {w_{ij} \choose 2}$, where $w_{ij}$ is the number of out-out and in-out wedges
% with $i$ and $j$ as endpoints. Starting from $i$, we can determine $w_{ij}$ be enumerating
% these wedges. No wedge is enumerated more than once, and the theorem holds.
% 
% The same idea can be directly adapted to get per-vertex and per-edge counts. 
% Consider some edge $(u,v)$, where $u \succ v$. We can express
% the number of four-cycles on edge $(u,v)$ by splitting into the three
% types of four-cycles, (a), (b), and (c). 
% 
% We have $count_a = \sum_{j \in \Gamma^+(v)} w^{++}_{uj}$, $count_b = \sum_{j \in \Gamma^-(v)} w^{+-}_{uj}
% + \sum_{i \in \Gamma^+(u)} w^{+-}_{iv}$, and $count_c$ \Sesh{work out the rest}
% \end{proof}
% 

\section{Onto 5-vertex counts} \label{sec:fivevertex}

With the cut framework of \Sec{cut}, we can generate
\emph{efficient} formulas for all 5-vertex patterns, barring the 5-cycle (pattern 8)
and the 5-clique (pattern 21). We give the formulas that \Lem{cut} yields. 
It is cumbersome and space-consuming to give proofs of all of these, so we omit them. 
We break the formulas into four groups, depending on whether the cut chosen in a vertex, edge, triangle, or wedge. 
We use $\tailedtri(G)$ to denote the tailed triangle count in $G$.
After stating these formulas, we will later explain the algorithm that
computes the various $N_i$s.
% (In all cases,
% it is possible to give an ad hoc proof of correctness.)

\begin{theorem} \label{thm:5cutvertex} [Cut is vertex]
$    N_1  =  \sum_i {d(i) \choose 4}$
 
$    N_3  =  \sum_i \sum_{(i,j) \in E} (d(j)-1) - 4\cdot C_4(G) - 2\cdot \tailedtri(G) - 3\cdot T(G)$
 
$    N_4  =  \sum_i t(i)(d(i)-2)$
 
$    N_7  =  \sum_i C_4(i)(d(i)-2) - 2\cdot \cc(G)$
 
$    N_9  =  \sum_i {t_i \choose 2} - 2 \cdot \cc(G)$
 
$    N_{15}  =  \sum_i K_4(i) (d(i)-3)$
\end{theorem}

\begin{theorem} \label{thm:5cutedge} [Cut is edge]
$    N_2  =  \sum_{\langle i,j\rangle \in E} (d(j)-1){d(i)-1 \choose 2} - 2\cdot \tailedtri(G) $

$    N_5  =  \sum_{e = \langle i,j\rangle \in E} (d(i)-1)(d(j)-T(e)) - 4\cdot \cc(G)$
 
$    N_6  =  \sum_{e = (i,j) \in E} t_e(d(i)-2)(d(j)-2) - 2\cdot \cc(G)$
 
$    N_{11}  =  \sum_{e = \langle i,j\rangle \in E} {T(e) \choose 2} (d(i)-3)$
 
$    N_{12}  =  \sum_{e \in E} C_4(e) T(e) - 4\cdot \cc(G)$
 
$    N_{14}  =  \sum_e {T(e) \choose 3}$
 
$    N_{19}  =  \sum_e K_4(e)(t(e)-2)$
\end{theorem}

For the next theorem, we give a short proof sketch of how the formulas are obtained.

\begin{theorem} \label{thm:5cuttri} [Cut is triangle]

$    N_{10}  =  \sum_{t = \langle i,j,k\rangle \ \textrm{triangle}} [(t(i,j)-1)(d(k)-1)] - 4\cdot K_4(G)$
 
$    N_{16}  =  \sum_{t = \langle i,j,k\rangle \ \textrm{triangle}} (t(i,j)-1)(t(i,k)-1)$
 
$    N_{20}  =  \sum_{t \ \textrm{triangle}} {K_4(t) \choose 2} - 4\cdot K_4(G) $
\end{theorem}

\begin{proof} Refer to \Fig{5vertex} for labels. We will
apply \Lem{cut}, where $C$ will be a triangle.
For pattern 10, we use vertices $\{1, 2, 4\}$ as $C$; for pattern 16, the cut is $\{2, 3, 4\}$;
for pattern 20, the cut is $\{3, 4, 5\}$. The formulas can be derived using \Lem{cut}.
\end{proof}

\begin{theorem} \label{thm:5cutwed} [Cut is pair or wedge] Define $\cc(i,j)$ to be the number
of diamonds involving $i$ and $h$ where $i$ and $j$ are not connected to the chord
(in \Fig{5vertex}, $i$ maps 1 and $j$ maps to $4$). Let $CC(i,j,k)$ be the number
of diamonds where $i$ maps to 1, $j$ maps to 2 and $k$ maps to 4.

$    N_{13}  =  \sum_{i \prec j} {W(i,j) \choose 3}$
 
$    N_{17}  =  \sum_{i \prec j} (W(i,j)-2)\cc(j,i)$
 
$    N_{18}  =  \sum_{i,j,k} {\cc(i,j,k) \choose 2} $
\end{theorem}

\begin{proof} The formula for $N_{13}$ is straightforward. For pattern 17, we choose vertices 3 and 4
as the cut. Observe that vertices 1, 2, 3, and 4 form a diamond. For the wheel (pattern 18), we
use the ``diagonal" 2, 1, 5 as the cut. The fragments are both diamonds sharing those vertices.
\qed
\end{proof}

Finally, we put everything together. The following theorem (and proof) show an algorithm
that uses the formulas given above.

\begin{theorem} \label{thm:5formula} Assume we have all the information from \Thm{4vertex-tri}. 
All counts in \Thm{5cutvertex}, \Thm{5cutedge}, and \Thm{5cuttri} can be computed in time 
$O(W(G) + \cc(G) + n + m)$ and $O(n+m+T(G)$ storage.
\end{theorem}

\begin{proof}  The counts of \Thm{5cutvertex}, \Thm{5cutedge}, and \Thm{5cuttri}
can be computed in $O(n)$, $O(m)$, and $O(T)$ time respectively. We can obviously
count $N_{13}$ in $O(W)$ time, by enumerating all wedges. For the remaining, we need
to generate $\cc(i,j)$ and $\cc(i,j,k)$ counts.

Let us describe the algorithm for $N_{17}$. Fix a vertex $i$. For every edge $(i,k)$, we have the list
of triangles incident to $(i,k)$ (from \Thm{4vertex-tri}). For each such triangle $(i,k,\ell)$,
we can get the list of triangles incident to $(k,\ell)$. For each such triangle $(k,\ell,j)$,
we have generated a diamond with $i$ and $j$ at opposite ends. By performing this enumeration
over all $(i,k)$, and all $(i,k,\ell)$, we can generate $CC(i,j)$ for all $j$.
By doing a 2-step BFS from $i$, we can also generate all $W(i,j)$ counts. Thus, we compute
the summand, and looping over all $i$, we compute $N_{17}$. The total running time
is the number of diamonds plus wedges. An identical argument holds for $N_{18}$ and is omitted.
\end{proof}

\subsection{The 5-cycle and 5-clique} \label{sec:5cycle}

The final challenge is to count the 5-cycle and the 5-clique. The main tool
is to use the DAG $G^\rightarrow$, analogous to 4-cycles and 4-cliques.

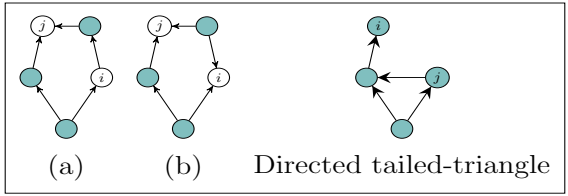
\begin{figure}[t]
\centering
\resizebox{3in}{1.0in}{
  \begin{tikzpicture}[nd/.style={scale=0.8,circle,draw,fill=teal!50,inner sep=2pt},minimum width = 13pt, framed]    
    \matrix[column sep=0.4cm, row sep=0.2cm,ampersand replacement=\&]
    {
     \FiveCycleOne \& 
     \FiveCycleTwo \&
     \DirTailedTri \\
%      \DirForClique \\
 };
  \end{tikzpicture}
  }  
\caption{Directed patterns for 5-cycle and 5-clique counting}
\label{fig:5cycle}
\end{figure}

\begin{theorem} \label{thm:5cycle} Consider the 3-path in \Fig{5cycle}, and let $P(i,j)$ be the
number of \emph{directed 3-paths} between $i$ and $j$, as oriented in the figure. 
Let $Z$ be the number of directed tailed-triangles, as shown in \Fig{5cycle}.
The number of 5-cycles
is $\sum_{i\prec j} P(i,j)\cdot (W_{++}(i,j) + W_{+-}(i,j)) - Z$.
\end{theorem}

\begin{proof} \Fig{5cycle} shows the different possible 5-cycle DAGs. There are only two (up to isomorphism).
In both cases, we choose $i$ and $j$ (as shown) to be the cut. 
(Wlog, we assume that $i \prec j$.)
The vertices have the same directed three-path between
them. They also have either an outwedge or inout-wedge connecting them.
Thus, the product $\sum_{i \prec j} P(i,j)\cdot (W_{++}(i,j) + W_{+-}(i,j))$ counts each 5-cycle exactly once.
The shrinkage of either directed 5-cycle yields the directed tailed-triangle of \Fig{5cycle}(d). 
This pattern is also counted exactly once in the product above. 
(One can formally derive this relation using \Lem{cut}.) Thus, we substract $Z$ out to get the number
of 5-cycles.
\end{proof}

\begin{theorem} \label{thm:5clique} (We remind the reader that $\dbp$ is the count of the
directed bipyramid in \Fig{5vertex-basis}.) The number of 5-cliques can be counted 
in time $O(\dbp(G) + \cc(G) + T(G) + n + m)$.
\end{theorem}

\begin{proof} First observe  that every 5-clique in $D$ contains one
of these directed bipyramids. Thus, it suffices to enumerate them to enumerate all 5-cliques.
The key is to enumerate this pattern with minimal overhead. From \Thm{4vertex-tri}, 
we have the list of triangles incident to every edge. For every triangle $t$, we determine
all of these patterns that contain $t$ as exactly the triangle $(i,j,k)$ in \Fig{5vertex-basis}.

Suppose triangle $t$ consists of vertices $i,j,k$.
We enumerate every other triangle incident to $j, k$ using the data structure of \Thm{4vertex-tri}.
Such a triangle has a third vertex, say $\ell$. We check if $i,j,k,\ell$ form the desired
directed configuration.
Once we generate all possible $\ell$ vertices, every pair among them gives the desired directed pattern.

The time required to generate the list of $\ell$ vertices over all triangles
is at most $\sum_{t = (i,j,k)} t(j,k) \leq \sum_{j,k} t(j,k)^2 = O(\cc(G) + T(G))$.
Once these lists are generated, the additional time is exactly $\dbp$ to generate each
directed pattern.
\end{proof}

At long last, we can prove \Thm{fivevertex}.

\begin{proof} (of \Thm{fivevertex}) We simply combine
all the relevant theorems: \Thm{4vertex-tri}, \Thm{5formula}, \Thm{5cycle},
and \Thm{5clique}. The runtime of \Thm{4vertex-tri} is $O(W_{++}(G^\rightarrow) + 
W_{+-}(G^\rightarrow) + \dcc(G^\rightarrow) + m + n)$.
The overhead of \Thm{4vertex-tri} is $O(W(G) + \cc(G) + m + n)$. Note that this dominates
the previous runtime, since it involves undirected counts. To generate $P(i,j)$ counts,
as in \Thm{5cycle}, we can easily enumerate all such three-paths from vertex $i$.
We can also generate $W(i,j)$ counts to compute the product, and the eventual sum.
Enumerating these three-paths will also find all of the directed tailed triangles of \Fig{5cycle}.
Thus, we pay an additional cost of $\dpath(G^\rightarrow)$. We add in the time of \Thm{5clique} to get
the main runtime bound of $O(W + \cc + \dpath(G^\rightarrow) + \dbp(G^\rightarrow) + m + n)$. 
The storage is dominated by \Thm{4vertex-tri}, since we explicitly store every triangle of $G$.
\end{proof}

\section{Experimental Results} 
We implemented our algorithms in {\tt C++} and ran our experiments on a
computer equipped with a 2x2.4GHz Intel Xeon processor with 6~cores
and  256KB  L2 cache (per core), 12MB L3 cache, and  64GB memory. 
We ran ESCAPE on a large 
collection of graphs from the Network Repository~\cite{network} and SNAP~\cite{Snap}.
In all cases, directionality is ignored, and duplicate edges and self loops  are omitted. 
\Tab{properties} has the properties of all these graphs.

The entire ESCAPE package is available as open source code (including
the code used in these results) at~\cite{escape}. 

% and per private communication
%with the authors of~\cite{MiMaGu07}.
%\footnote{http://socialnetworks.mpi-sws.org/}. 

%
%\input{WWWFigs/properties}
%
\begin{table*}
\scriptsize
\begin{center}
\caption{\label{tab:properties} Properties of the graphs.} 
\begin{tabular}{l|rrr|rrr|}
\multicolumn{1}{c|}{} & \multicolumn{3}{c|}{} & \multicolumn{3}{c|}{Runtimes in seconds}\\
 & 	$|V|$	& $|E|$	&  $|T|$ & PGD & ESC-4  & ESC-5 \\\hline 
% $\frac{W_{++}}{W}$&$\frac{W_{+-}}{W}$&T	&3-star & 3-path& tailed & 	4-cycle & Chordal&  4-clique & Time \\ \multicolumn{1}{c|}{} & \multicolumn{2}{c|}{} & \multicolumn{3}{c|}{} &\multicolumn{2}{c}{}& \multicolumn{1}{c}{triangle} & \multicolumn{1}{c}{}& \multicolumn{1}{c}{4-cycle} &\multicolumn{1}{c|}{}  & \multicolumn{1}{c|}{(secs)}\\ \hline 
soc-brightkite & 	56.7K & 	426K	 & 494K	 & 1.20	 & 0.22	 & 6.54 \\
tech-RL-caida & 	191K & 	1.22M	 & 455K	 & 3.21	 & 0.25 & 	5.47 \\
flickr	      &         244K & 	3.64M	 & 15.9M & 	809K	 & 12.9 & 	961K \\
ia-email-EU-dir & 	265K & 	729K & 	267K & 	10.6 & 	0.18	 & 8.69 \\
ca-coauth-dblp & 	540K	 & 3.05M	 & 444M & 	585 & 	615 & 	47.4K \\
web-google-dir & 	876K	 & 8.64M	 & 13.4M & 	54.5 & 	2.94	 & 71.8 \\
tech-as-skitter & 	1,69M & 	22.2M	 & 28.8M	 & 1.90K	 & 20.3	 & 1.41K \\
web-wiki-ch-int & 	1.93M	 & 9.16M	 & 2.63M	 & 4.91K	 & 6.80 & 	798 \\
web-hudong      & 	1.98M	 & 14.6M & 	5.07M	 & 9.40K & 	13.6 & 	534 \\
wiki-user-edits & 	2.09M	 & 11.1M	 & 6.68M & 	439K	 & 2.92	& 9.15K \\
web-baidu-baike & 	2.14M	 & 17.4M & 	3.57M	 & 22.9K & 	16.2 & 	 9.46K \\
tech-ip     &   	2.25M	 & 21.6M & 	298K & 	613K	 & 25.7	 & 295 \\
orkut	      &        3.07M & 	234M & 	628M & 	598K	 & 1.19K & 	217K \\
LiveJournal & 	        4.84M	 & 85.7M	 & 286M	 & 25.9K	 & 538	 & 37.1K \\
\hline
\end{tabular}
\end{center}
\end{table*}

%\begin{figure*}[tbp]
 % \centering
 % \small
 % \begin{tabular}{ccc}
  %\includegraphics[height=1.0in]{WWWFigs/Induced-to-noninduced.png} 
%& \includegraphics[height=1.0in]{WWWFigs/transitivity.png} \; 
%& \includegraphics[height=1.0in]{WWWFigs/Nowheels.png} \\
% (a) & (b) & (c) 
% \end{tabular}
  %\subfloat[ L
 % {\label{fig: inducedratio}\includegraphics[width=0.2\textwidth]{WWWFigs/induced-to-noninduced.png}}
 % \subfloat[Clustering coefficients for ca-HepTh]
 % {\label{fig:ca-HepTh-cc}\includegraphics[width=0.2\textwidth]{WWWFigs/transitivity.png}}
 % \subfloat[Leading adjacency matrix eigenvalues for ca-HepTh]
 % {\label{fig:ca-HepTh-eigs}\includegraphics[width=0.3\textwidth]{WWWFigs/nowheels.png}}
 % \caption{\small Trends in pattern counts. (a) ikelihood that a copy of pattern-i contains another edge, measured as $1- C_i/N_i$ across all graphs. Patterns are labeled as i-j for j-th i-pattern. 3-1 refers to a two path. (b) Comparing transitivity of 3 common neighbors and  2 common neighbors (c) instability of the wheel pattern; comparing how often patterns 18 and 19 turbot pattern 20.  }
  %\label{fig:trends}
%\end{figure*}

\medskip 

\textbf{4-vertex pattern counting:}  We compare ESCAPE with the
Parallel Parameterized Graphlet Decomposition (PGD) Library~\cite{AhNe+15}. 
PGD can exploit parallelism using multiple threads. But our focus
is on the basic algorithms, so we ran ESCAPE and PGD on a single thread.
%  but our focus is on  the algorithms, and thus to compare the algorithmic techniques we ran PGD on a single thread (our algorithms can be parallelized the same way as in PGD, but this is not yet implemented in the code).  
% 
The runtimes of PGD and ESCAPE are given in \Tab{properties}, 
 and the speedups, computed as  ratio of  PGD runtime to  ESCAPE runtime,  are presented in \Fig{4speedup}.
PGD has not completed after over 170 and 121 hours for tech-ip and ia-wiki-user-edits graphs, respectively and thus we use these times  as lower bounds for runimtes of PGD for these graphs.  
The only instance where PGD was faster is {\tt ca-coauthors-dblp}, where
the runtimes were comparable.
In almost all medium sized instances (< 10M edges), 
we observe a one order of magnitude of speedup on medium sized instances.
For large instances (100M edges), ESCAPE gives two orders of magnitude speedup 
over PGD. For instance on the orkut graph with 234M edges, ESCAPE runs more than 500 times faster than PGD.  
We  should also note that  PGD is already a well-designed code  based on strong algorithms.  Most notably, overall runtimes are in the order of seconds for these very large graphs, as displayed on the right most column in \Tab{properties}.  For instance,  computing exact counts on the {\tt as-skitter} graph with 1.7M vertices and 11.1M edges took only 21.79 seconds. 
We assert that exact 4-vertex pattern counting is quite feasible,
with reasonable runtimes, for even massive graphs. We present  counts of all 5-patterns in \Sec{fullcounts} of the Appendix.

% \begin{figure}[th]
% 	        \centering
%     	\includegraphics[width=0.8\columnwidth]{4vertex-Speedup.png}
%    	\caption{\label{fig:4speedup} Speedup achieved by Escape over PGD (computed as runtime of ESCAPE/runtime of PGD). }
%   
%   \end{figure}
% \Fig{4speedup} shows the speedup achieved by ESCAPE over PGD. Note that the vertical axis in this figure uses logarithmic scale.  In our  experiments, PGD failed to complete in two graphs, tech-ip and  wiki-user-edits-page after 
% 658K seconds, and we stopped these processes and used this time as the completion time.   The figure shows that  ESCAPE offers  significant speedups over PGD. The only exception is {\tt ca-coauthors-dblp}, where PGD was 8\% faster than ESCAPE. One of the algorithmic insights of ESCAPE is using the edge orientation  as explained earlier, but as \Tab{properties} shows, edge orientation  offers little in this case ($W_{++}$ and $W_{+-}$ wedges make up a significant portion of all wedges), which explains the  ESCAPE  performance  for this graph.  

\medskip 

\textbf{5-vertex Pattern Counting:}  ESCAPE runtimes for counting  5-patterns are also presented in \Tab{properties}.
We note that 5-vertex pattern counting can be done in minutes for graphs with less
than 10M edges. For instance, ESCAPE computes all 5-patterns for  tech-ip  with 2.25M nodes and 21.6M edges in less than 5  minutes. Thus, randomization is quite unnecessary for graphs of such size.
No other method we know of can handle even such medium size graphs for this problem.
It is well-documented (refer to~\cite{JhSePi15} for an analysis
of 4-cliques, and to~\cite{AhNe+15} for comparisons to PGD) 
that existing methods cannot scale for 10M edge graphs:
FANMOD~\cite{WernickeRasche06}, edge sampling methods~\cite{TsKaMiFa09,RaBhHa14}, 
ORCA~\cite{orca13}. %\Sesh{more on our results}

\begin{figure}[h]
\centering
 \includegraphics[width=0.9\columnwidth]{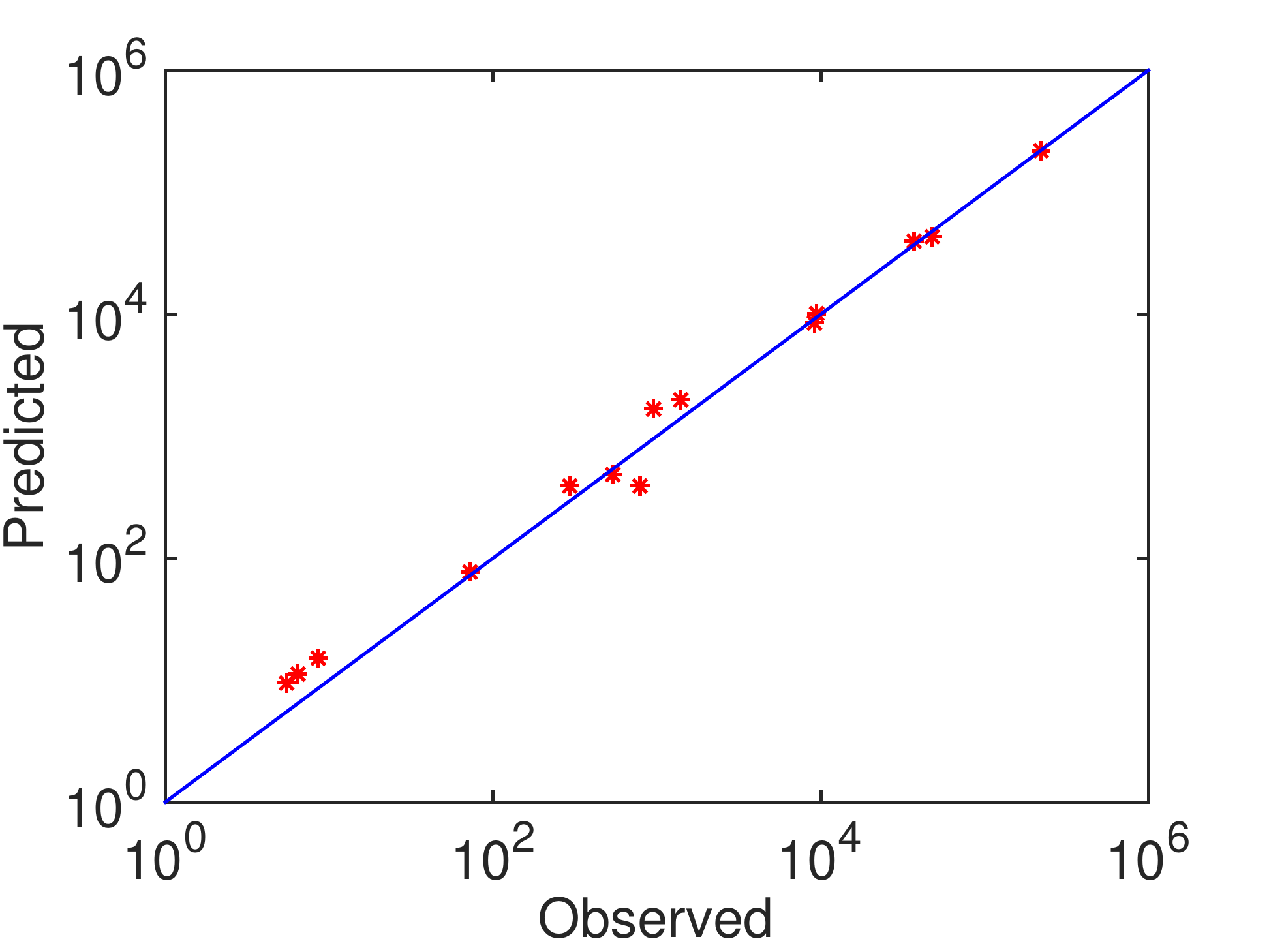}
 \caption{\label{fig:predict} \small Predicting ESCAPE runtime for 5-patterns.  The line is defined by  $1.0E-4*(1.39W(G)+1.09cc(G)+ 24.28\dpath(G^\rightarrow)+4.41\dbp(G^\rightarrow))$. } 
 \end{figure}

\textbf{Runtime Predictions}  
%of \Thm{fourvertex} and \Thm{fivevertex}:} \
\Thm{fourvertex} and \Thm{fivevertex} claim that the runtime of the ESCAPE algorithm is
bounded by the counts of specific patterns (as shown in~\Fig{5vertex-basis}).
Here we present our validation for only 5-patterns due to space.  We fit a line 
using coefficients for $W(G)$, $\cc(G)$, $\dpath(G^\rightarrow)$ and $\dbp(G^\rightarrow)$. We do not use  $m$ and $n$ to limit degrees of freedom.  
And the result is presented in \Fig{predict}, which shows that the runtime can be accurately predicted as a function of  counts for base patterns as described in  \Thm{fivevertex}.
%$O(W(G) + \cc(G) + \dpath(G^\rightarrow) + \dbp(G^\rightarrow) + m + n)$ 
%we in~\Fig{5vertex-basis} versus the wall-clock time of the ESCAPE. 

\textbf{Trends in pattern counts:} We analyze the actual counts of the various
patterns, and glean the following trends. 
\begin{asparaitem}
    \item I{\it nduced vs non-induced:} For all patterns, we look
    at the ratio $C_i/N_i$, the fraction of non-induced matches of a pattern
    that are also induced. Conversely, one can interpret $1- C_i/N_i$
    as the ``likelihood" that a copy of pattern-i contains another edge.  We present the results in \Fig{inducedratio}.
    The surprising observation that across all the graphs, certain
    patterns are extremely rarely induced. It is extremely infrequent to observe 5-patterns 16--20 as induced patterns, which can be good tool for edge prediction.  Note that  wedges/triangles are commonly used for edge prediction, but across all graphs wedge-to-triangle closure is not frequent at all. This ratio ca be high for some graphs, for an arbitrary graph wedges by themselves are not good edge predictors.  Also note that 5-pattern 1 frequently  remains as a induced pattern, even though it has 6 potential missing edges.  These results show that going beyond 3 and 4-vertex patterns  can reveal more interesting structures and provide predictive power.

 %\Sesh{say more, maybe some relevance to link prediction}

\begin{figure}[th]
\centering
\includegraphics[width=0.7\columnwidth]{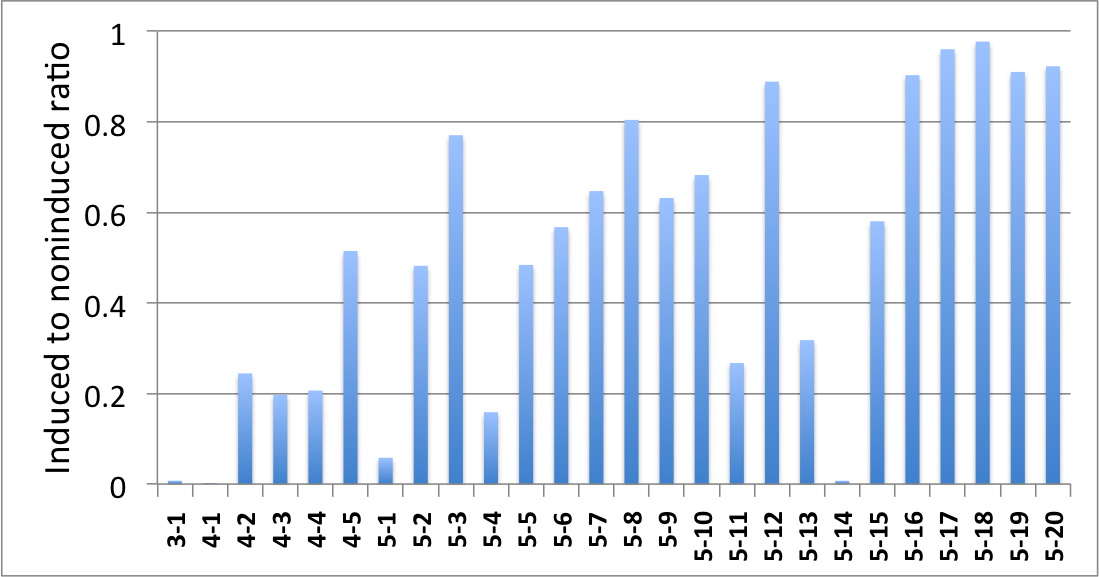}
 \caption{\label{fig:inducedratio} \small Likelihood that a copy of pattern-i contains another edge, measured as $1- C_i/N_i$ across all graphs. Patterns are labeled as i-j for j-th i-pattern. 3-1 refers to a two path.  } 
 \end{figure}

    \item {\it A measure of transitivity}: What is the likelihood that to vertices
    with two neighbors are connected by an edge? What if it was three neighbors
    instead? An alternate (not equivalent) method to measure this is the see the fraction of 
    4-cycles that form diamonds, and the fraction of (13) that form (14).
    (The latter is basically taking a pattern where two vertices have three neighbors,
    and see how often those vertices have an edge.) \Fig{transitivity} shows that  having 3 common neighbors significantly increases likelihood of an edge, especially for  social networks. 
\begin{figure}[th]
\centering
 \includegraphics[width=0.9\columnwidth]{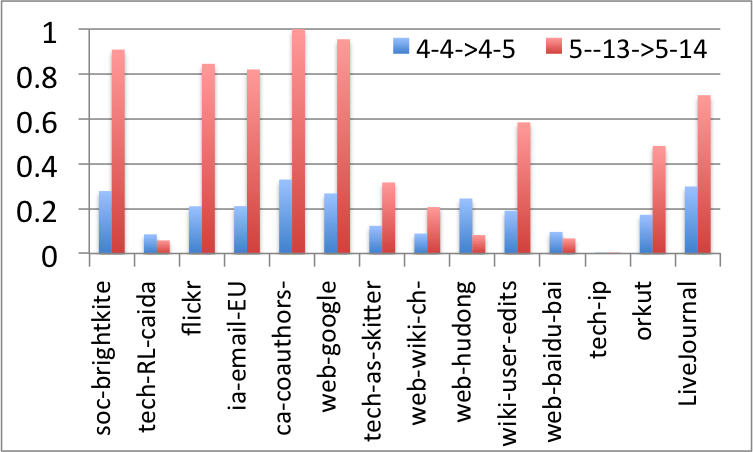}
 \caption{\label{fig:transitivity} Comparing transitivity of 3 common neighbors and  2 common neighbors } 
\end{figure}

    \item {\it The lack of wheels}: The intriguing fact is that wheels (pattern (18))
    are much rarer that one would expect. It appears to be an ``unstable" pattern. \Fig{inducedratio} already shows that they are infrequent as induced patterns. Here we will go a step further and how often pattern 18 has an additional edge to turn into pattern 20, and as a basis for comparison we will compare it with that of pattern 19, which has the same number of edges. Results are presented in \Fig{wheel}, which shows that P18 is more than twice as likely to tun into Pattern 20 compared to  pattern 19. 
 \begin{figure}[th]
\centering
\includegraphics[width=0.9\columnwidth]{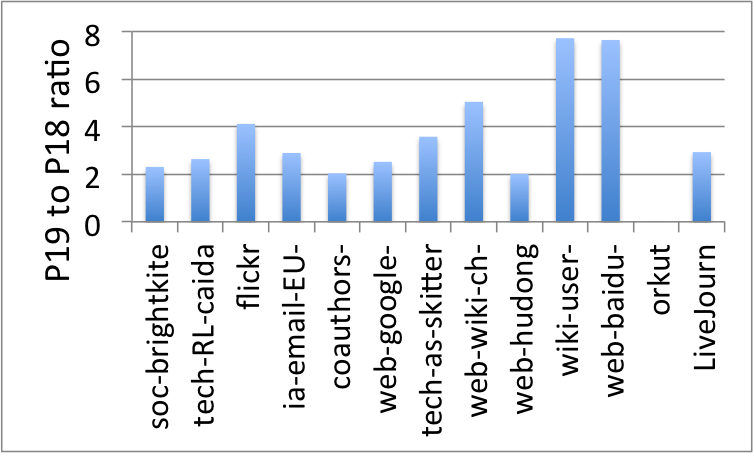}
\caption{\label{fig:wheel} Ratio of  5-pattern 19 to 5-pattern 18.  } 
\end{figure}
\end{asparaitem}

\bibliographystyle{abbrv}
\bibliography{motifcount}

\appendix
\section{Disconnected patterns} \label{sec:disc}

Disconnected patterns can be easily counted
in extra \emph{constant} time, if we have all connected patterns
up to a given size. 

\begin{theorem} Fix a graph $G$. Suppose we have
counts for all connected $r$-vertex patterns, for all $r \leq k$.
Then, the counts for all (even disconnected) $k$-vertex patterns
can be determined in constant time (only a function of $k$). 
\end{theorem} 

\begin{proof} Basically, we show that the counts of disconnected
patterns is just a polynomial function of counts of
connected patterns with fewer vertices. Formally, we prove
by induction on $c$, the number of connected components in the pattern.
The base case, $c=1$, is trivially true by assumption. For induction, suppose
we have counts for all patterns (of at most $k$ vertices) 
with at most $c$ connected components.
Consider pattern $H$ with connected components $H_1, H_2, \ldots, H_{c+1}$.
The count is the number of 1-1 matches $f: V(H) \to V(G)$ mapping
edges of $H$ to edges of $G$. Observe that the function ``partitions"
into 1-1 functions $f_i: V(H_i) \to V(G)$ ($\forall i \leq c+1$)
such that all images are disjoint. But the number of $f_i$s is
exact the number of matches of $H_i$. We now use a one-step 
inclusion-exclusion to determine $\match(H)$.
$$\match(H) = \prod_{i=1}^{c+1} \match(H_i) - |\cF|$$
where $\cF$ is the set of maps $f:V(H) \to V(G)$ such that
each $f|_{V(H_i)}$ is 1-1, but $f$ is not 1-1. This
can only happen if $f$ maps to a single vertex some set $S \subseteq V(H)$
that is \emph{not} contained in a single $V(H_i)$.

We now partition $\cF$ according to the following scheme.
Consider a non-trivial partition $\cS$ of $V(H)$ with the following
condition: a non-singleton in $\cS$ cannot be contained in some
$V(H_i)$. Define $\cF_\cS$ to be the set of maps
where each $S \in \cS$ is mapped to a single vertex, but these
vertices are different for distinct $\cS$. Observe that for any
$f \in \cF_\cS$, $f|_{V(H_i)}$ is 1-1 for all $i \leq c+1$.
Thus, $\cF$ is partitioned into the $\cF_\cS$s, and $|\cF| = \sum_{\cS} |\cF_\cS|$.

Crucially, we note that $\cF_\cS$ is exactly the set of matches
for the pattern $H_\cS$ created by merging $V(H)$
as follows: for each $S \in \cS$, merge $S$ into a single vertex.
Since $\cS$ is non-trivial, it contains some non-singleton $S$.
This non-singleton is not contained in any $V(H_i)$; thus,
merging $S$ reduces the number of connected components in $H$.
Thus, $H_{\cS}$ has strictly less than $c+1$ connected components,
and by induction, we already know its count.

Thus, we write our main equation as follows, and observe
that the right hand size involves counts that are already know
(by induction).
$$\match(H) = \prod_{i=1}^{c+1} \match(H_i) - \sum_{\cS} \match(H_\cS)$$
\end{proof}

\begin{figure}[t]
\centering
\resizebox{3in}{1.4in}{
  \begin{tikzpicture}[nd/.style={circle,draw,fill=teal!50,inner sep=2pt},framed]    
    \matrix[column sep=0.4cm, row sep=0.2cm,ampersand replacement=\&]
    {
     \IndSetFour \&    
     \EdgeFour \&  
     \MatFour  \\
     \WedgeFour  \& 
     \TriFour     \\
 };
  \end{tikzpicture}  
  }
\caption{Disconnected 4-vertex patterns}
\label{fig:4vertex-disc}
\end{figure}
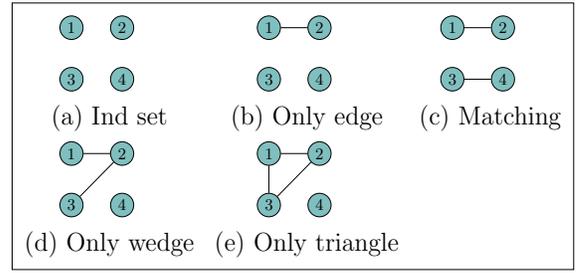

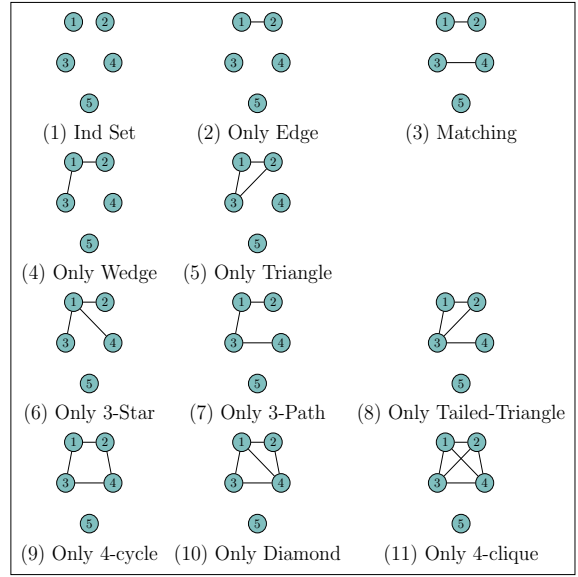
\begin{figure}[t]
\centering
\resizebox{3in}{3.0in}{
  \begin{tikzpicture}[nd/.style={circle,draw,fill=teal!50,inner sep=2pt},framed]    
    \matrix[column sep=0.4cm, row sep=0.2cm,ampersand replacement=\&]
    {
     \IndSetFive \&    
     \EdgeFive \&  
     \MatFive  \\
     \WedgeFive  \& 
     \TriFive    \\
     \ThreeStarFive \&
     \ThreePathFive \&
     \TailedTriFive \\
     \FourCycleFive \&
     \DiamondFive \&
     \FourCliqueFive \\
 };
  \end{tikzpicture}  
  }
\caption{Disconnected 5-vertex patterns}
\label{fig:5vertex-disc}
\end{figure}

\section{Conversion between Induced and Noninduced counts}
\label{sec:transform}
We show how to convert induced pattern counts to non-induced
counts, and vice versa. Algorithmically, it is often easier
to get non-induced counts, and a linear transformation of these
counts suffices to get induced counts.
It is standard that each non-induced
pattern count can be obtained as a linear combination of
induced pattern counts (of that size). 
We use $\ind_i$ (resp. $\nonind_i$) to denote the induced (resp. non-induced)
count of the $i$th $5$-pattern. We can think of the list
of counts as vectors $\ind, \nonind$.
There is a matrix $\bA$ such that for all graphs, $\nonind = \bA \ind$.
The matrix is given in \Fig{A}. It is constructed by noting that $\bA_{i,j}$ is the number 
of copies of pattern $j$ in pattern $i$. 
Naturally, $\ind = \bA^{-1} \nonind$, and that is how we get induced counts
from non-induced counts. The inverse matrix is given in~\Fig{Ainv}.

\section{Detailed induced  counts}\label{sec:fullcounts} 
 Here we present detailed counts for  various patterns unto 5 vertices.  First, \Tab{prop-full} presents all connected 4-patterns and some of the 3-vertex patterns relevant to our algorithms.  
 We also report the induced counts  for all  connected 5-vertex patterns in \Tab{5vertex1} and \Tab{5vertex2}.   
 
\begin{figure*}
\[ A=
\left(
\begin{array}{ccccccccccccccccccccc}

1 &   &   & 1 &   &   &   &   & 1 &   & 1 &   &   & 2 & 1 & 1 &   & 1 & 2 & 3 & 5 \\ 
  & 1 &   & 2 & 1 & 2 & 2 &   & 4 & 4 & 5 & 4 & 6 & 12 & 9 & 10 & 10 & 20 & 20 & 36 & 60 \\ 
  &   & 1 &   & 2 & 1 & 2 & 5 & 4 & 4 & 2 & 7 & 6 & 6 & 6 & 10 & 14 & 24 & 18 & 36 & 60 \\ 
  &   &   & 1 &   &   &   &   & 2 &   & 2 &   &   & 6 & 3 & 3 &   & 4 & 8 & 15 & 30 \\ 
  &   &   &   & 1 &   &   &   & 4 & 2 &   & 2 &   &   & 3 & 6 & 6 & 16 & 12 & 30 & 60 \\ 
  &   &   &   &   & 1 &   &   &   & 2 & 2 & 1 &   & 6 & 6 & 5 & 4 & 12 & 14 & 30 & 60 \\ 
  &   &   &   &   &   & 1 &   &   & 1 & 1 & 2 & 6 & 6 & 3 & 4 & 8 & 16 & 12 & 30 & 60 \\ 
  &   &   &   &   &   &   & 1 &   &   &   & 1 &   &   &   & 1 & 2 & 4 & 2 & 6 & 12 \\ 
  &   &   &   &   &   &   &   & 1 &   &   &   &   &   &   & 1 &   & 2 & 2 & 6 & 15 \\ 
  &   &   &   &   &   &   &   &   & 1 &   &   &   &   & 3 & 2 & 2 & 8 & 8 & 24 & 60 \\ 
  &   &   &   &   &   &   &   &   &   & 1 &   &   & 6 & 3 & 2 &   & 4 & 10 & 24 & 60 \\ 
  &   &   &   &   &   &   &   &   &   &   & 1 &   &   &   & 2 & 4 & 12 & 6 & 24 & 60 \\ 
  &   &   &   &   &   &   &   &   &   &   &   & 1 & 1 &   &   & 1 & 2 & 1 & 4 & 10 \\ 
  &   &   &   &   &   &   &   &   &   &   &   &   & 1 &   &   &   &   & 1 & 3 & 10 \\ 
  &   &   &   &   &   &   &   &   &   &   &   &   &   & 1 &   &   &   & 2 & 6 & 20 \\ 
  &   &   &   &   &   &   &   &   &   &   &   &   &   &   & 1 &   & 4 & 4 & 18 & 60 \\ 
  &   &   &   &   &   &   &   &   &   &   &   &   &   &   &   & 1 & 4 & 1 & 9 & 30 \\ 
  &   &   &   &   &   &   &   &   &   &   &   &   &   &   &   &   & 1 &   & 3 & 15 \\ 
  &   &   &   &   &   &   &   &   &   &   &   &   &   &   &   &   &   & 1 & 6 & 30 \\ 
  &   &   &   &   &   &   &   &   &   &   &   &   &   &   &   &   &   &   & 1 & 10 \\ 
  &   &   &   &   &   &   &   &   &   &   &   &   &   &   &   &   &   &   &   & 1 \\ 
\end{array}
\right)
\]
\caption{Matrix transforming induced 5-vertex pattern counts to non-induced counts}
\label{fig:A}
\end{figure*}

\begin{figure*}\[ 
A^{-1}=
\left(
\begin{array}{rrrrrrrrrrrrrrrrrrrrr}
1 &   &   & -1 &   &   &   &   & 1 &   & 1 &   &   & -2 & -1 & -1 &   & 1 & 2 & -3 & 5 \\ 
  & 1 &   & -2 & -1 & -2 & -2 &   & 4 & 4 & 5 & 4 & 6 & -12 & -9 & -10 & -10 & 20 & 20 & -36 & 60 \\ 
  &   & 1 &   & -2 & -1 & -2 & -5 & 4 & 4 & 2 & 7 & 6 & -6 & -6 & -10 & -14 & 24 & 18 & -36 & 60 \\ 
  &   &   & 1 &   &   &   &   & -2 &   & -2 &   &   & 6 & 3 & 3 &   & -4 & -8 & 15 & -30 \\ 
  &   &   &   & 1 &   &   &   & -4 & -2 &   & -2 &   &   & 3 & 6 & 6 & -16 & -12 & 30 & -60 \\ 
  &   &   &   &   & 1 &   &   &   & -2 & -2 & -1 &   & 6 & 6 & 5 & 4 & -12 & -14 & 30 & -60 \\ 
  &   &   &   &   &   & 1 &   &   & -1 & -1 & -2 & -6 & 6 & 3 & 4 & 8 & -16 & -12 & 30 & -60 \\ 
  &   &   &   &   &   &   & 1 &   &   &   & -1 &   &   &   & 1 & 2 & -4 & -2 & 6 & -12 \\ 
  &   &   &   &   &   &   &   & 1 &   &   &   &   &   &   & -1 &   & 2 & 2 & -6 & 15 \\ 
  &   &   &   &   &   &   &   &   & 1 &   &   &   &   & -3 & -2 & -2 & 8 & 8 & -24 & 60 \\ 
  &   &   &   &   &   &   &   &   &   & 1 &   &   & -6 & -3 & -2 &   & 4 & 10 & -24 & 60 \\ 
  &   &   &   &   &   &   &   &   &   &   & 1 &   &   &   & -2 & -4 & 12 & 6 & -24 & 60 \\ 
  &   &   &   &   &   &   &   &   &   &   &   & 1 & -1 &   &   & -1 & 2 & 1 & -4 & 10 \\ 
  &   &   &   &   &   &   &   &   &   &   &   &   & 1 &   &   &   &   & -1 & 3 & -10 \\ 
  &   &   &   &   &   &   &   &   &   &   &   &   &   & 1 &   &   &   & -2 & 6 & -20 \\ 
  &   &   &   &   &   &   &   &   &   &   &   &   &   &   & 1 &   & -4 & -4 & 18 & -60 \\ 
  &   &   &   &   &   &   &   &   &   &   &   &   &   &   &   & 1 & -4 & -1 & 9 & -30 \\ 
  &   &   &   &   &   &   &   &   &   &   &   &   &   &   &   &   & 1 &   & -3 & 15 \\ 
  &   &   &   &   &   &   &   &   &   &   &   &   &   &   &   &   &   & 1 & -6 & 30 \\ 
  &   &   &   &   &   &   &   &   &   &   &   &   &   &   &   &   &   &   & 1 & -10 \\ 
  &   &   &   &   &   &   &   &   &   &   &   &   &   &   &   &   &   &   &   & 1 \\
\end{array}
\right)
\]
\caption{Matrix transforming non-induced 5-vertex pattern counts to induced counts}
\label{fig:Ainv}
\end{figure*}

 \begin{table*}
\scriptsize
\caption{\label{tab:prop-full} 4 and 5-vertex pattern (induced) counts for various graphs} 
\begin{tabular}{l|cc|ccc|cccccc|}
 & 	V	& E	& $\frac{W_{++}}{W}$&$\frac{W_{+-}}{W}$&T	&3-star & 3-path& tailed & 	4-cycle & Chordal&  4-clique \\ 
 \multicolumn{1}{c|}{} & \multicolumn{2}{c|}{} & \multicolumn{3}{c|}{} &\multicolumn{2}{c}{}& \multicolumn{1}{c}{triangle} & \multicolumn{1}{c}{}& \multicolumn{1}{c}{4-cycle} &\multicolumn{1}{c|}{}  \\ \hline 
soc-brightkite & 5.7e+4 & 4.3e+5 &0.076 & 0.194& 4.9e+5 & 1.3e+9 & 5.3e+8 & 1.1e+8 & 2.7e+6 & 1.2e+7 & 2.9e+6 \\
tech-RL-caida & 1.9e+5 & 1.2e+6 &0.082 & 0.173& 4.5e+5 & 1.7e+9 & 5.8e+8 & 7.7e+7 & 4e+7 & 7.4e+6 & 4.2e+5   \\
flickr  & 	2.4e+5& 	3.6e+6 &0.013 &0,029 & 1.6e+7& 1.1e+13	& 4.9e+11	& 1.4e+11	& 2.4e+09& 	4.7e+09	& 1.2e+08 \\
ia-email-EU-dir & 2.7e+5 & 7.3e+5 &0.005 & 0.023& 2.7e+5 & 2.2e+11 & 4.4e+9 & 3.4e+8 & 6.7e+6 & 1e+7 & 5.8e+5   \\
ca-coauthors-dblp & 5.4e+5 & 3e+7 &0.232 & 0.313& 4.4e+8 & 2.7e+10 & 4.2e+10 & 6.7e+10 & 3.1e+7 & 3.4e+9 & 1.5e+10 \\
web-google-dir & 8.8e+5 & 4.7e+6 &0.025 & 0.035& 2.5e+6 & 8.3e+10 & 1.5e+9 & 5.1e+8 & 1.4e+7 & 2.4e+7 & 2e+6  \\
tech-as-skitter & 1.7e+6 & 2.2e+7 &0.006 & 0.013& 2.9e+7 & 9.6e+13 & 8.2e+11 & 1.6e+11 & 4.3e+10 & 2e+10 & 1.5e+8  \\
web-wiki-ch-internal & 1.9e+6 & 1.8e+7 &0.015 & 0.047& 1.8e+7 & 3.1e+13 & 1.3e+12 & 1.3e+11 & 4.2e+9 & 2.1e+9 & 3e+7  \\
web-hudong & 2e+6 & 1.5e+7 &0.008 & 0.015& 5.1e+6 & 1.5e+13 & 2.2e+11 & 6.3e+9 & 2.8e+8 & 2.5e+8 & 8.3e+7  \\
wiki-user-edits-page & 2.1e+6 & 1.1e+7 &0.000 & 0.000& 6.7e+6 & 8.8e+16 & 4.8e+12 & 2e+12 & 4.4e+10 & 7e+10 & 1e+7  \\
web-baidu-baike & 2.1e+6 & 1.7e+7 &0.008 & 0.021& 3.6e+6 & 7.2e+13 & 4.8e+11 & 2.3e+10 & 5.9e+8 & 1.8e+8 & 7.1e+5  \\
tech-ip & 2.3e+6 & 2.2e+7 &0.001 & 0.000& 3.0e+5 & 1.3e+17 & 1.4e+13 & 8.6e+9 & 8e+11 & 3e+7 & 2  \\
orkut&3.1e+6	& 2.3e+8 &0.089 & 0.188& 6.3e+8&   9.8e+13	& 1.9e+13	& 1.5e+12& 	7.0e+10	& 4.8e+10& 	3.2e+9 \\
LiveJournal & 4.8e+6& 	8.6e+7& 0.094 & 0.190 &2.9e+8  &6.6e+12& 	1.1e+12& 1.2e+11	& 5.0e+9	& 1.7e+10	& 9.9e+9 \\ \hline
%
%ia-dbpedia-team-bi & 3.7e+5 & 8.4e+5 &0.011 & 0.010& 5.2e+2 & 3.9e+9 & 1.9e+8 & 3.5e+5 & 5.5e+5 & 2e+3 & 0   & 0.17\\
% web-wikipedia2009 & 1.9e+6 & 9e+6 &0.064 & 0.175& 2.2e+6 & 1e+10 & 4.3e+9 & 4e+8 & 6.4e+7 & 3e+7 & 1.5e+6  & 3.76\\
% ia-wiki-Talk-dir & 2.4e+6 & 4.8e+6 &0.004 & 0.025& 2e+6 & 2.4e+13 & 2e+11 & 7.2e+9 & 1.6e+8 & 9.6e+7 & 3.9e+6  & 28.58 \\
% roadNet-PA & 1.1e+6 & 3.1e+6 &0.217 & 0.398& 6.7e+4 & 1.4e+6 & 6.2e+6 & 2.9e+5 & 1.5e+5 & 5.7e+3 & 16 &0.23  \\
% roadNet-CA & 2e+6 & 5.5e+6 &0.215 & 0.402& 1.2e+5 & 2.4e+6 & 1.1e+7 & 5.2e+5 & 2.5e+5 & 1.3e+4 & 40  & 0.41\\
% road-usa & 2.4e+7 & 5.8e+7 &0.212 & 0.345& 4.4e+5 & 1.8e+7 & 8.1e+7 & 1.5e+6 & 1.6e+6 & 2.1e+4 & 90  & 5.26\\
% italy-osm & 6.7e+6 & 1.4e+7 &0.098 & 0.692& 7.4e+3 & 9.9e+5 & 9.9e+6 & 2.7e+4 & 4.7e+4 & 2.4e+2 & 0  & 0.77\\
% flickr & 1.7e+6 & 3.6e+6 &0.013 & 0.029& 1.6e+7 & 1.1e+13 & 4.9e+11 & 1.4e+11 & 2.4e+9 & 4.7e+9 & 1.2e+8  &3.02\ \\
% orkut & 3.1e+6 & 3.6e+7 &0.016 & 0.034& 6.1e+7 & 7e+13 & 1.5e+12 & 2.2e+11 & 8.4e+9 & 7.4e+9 & 3.3e+8  & 57.74\\
\hline
\end{tabular}
\end{table*}

\begin{table*}
\caption{\label{tab:5vertex1} Induced counts  for  5 vertex patterns;  Patterns 1--12 } 
\scriptsize
\begin{tabular}{l|ccccccccccc|}
& P1& P2& P3& P4& P5& P6& P7& P8& P9& P10& P11\\ \hline 
%roadNet-PA & 2.2e+5 & 7.5e+6 & 1.2e+7 & 9.2e+4 & 5.7e+5 & 4.2e+5 & 9e+5 & 8.5e+4 & 1.1e+4 & 1.7e+4 & 9.5e+3  \\
%roadNet-CA & 3.5e+5 & 1.3e+7 & 2e+7 & 1.6e+5 & 1e+6 & 7.2e+5 & 1.4e+6 & 1.5e+5 & 2e+4 & 3.6e+4 & 2e+4  \\
%road-usa & 2.4e+6 & 8.4e+7 & 1.3e+8 & 2.9e+5 & 2.4e+6 & 1.7e+6 & 9.4e+6 & 7.1e+5 & 1.6e+4 & 4.7e+4 & 2.2e+4  \\
soc-brightkite & 2.3e+11 & 1.3e+11 & 2.2e+10 & 1.7e+10 & 4e+9 & 9.7e+9 & 1.2e+9 & 3.5e+7 & 2.2e+8 & 1.2e+9 & 2.7e+9  \\
tech-RL-caida & 2.4e+11 & 8e+10 & 2.1e+10 & 8.8e+9 & 2.1e+9 & 5.1e+9 & 1.1e+10 & 3.4e+7 & 7e+7 & 5.1e+8 & 1.8e+9  \\
flickr &4.6e+16 	& 2.2e+15	& 1.3e+14	& 6.7e+14& 	8.9e+12& 	2.1e+14& 	2.0e+13	& 2.0e+11	& 7.4e+11	& 4.4e+12& 	5.8e+13\\
ia-email-EU-dir & 3.1e+14 & 4.3e+12 & 3.9e+11 & 1.9e+11 & 1.5e+10 & 1.4e+11 & 1e+10 & 1.4e+8 & 2e+8 & 4.7e+9 & 1.1e+10  \\
ca-coauthors-dblp & 5.7e+12 & 5e+12 & 2.8e+12 & 4.9e+12 & 4.2e+12 & 4.2e+12 & 1.1e+10 & 9.8e+8 & 1.7e+12 & 4.2e+11 & 3.6e+11  \\
web-google-dir & 4e+13 & 1.1e+12 & 4.9e+10 & 3.6e+11 & 3.5e+9 & 2e+10 & 7.3e+9 & 2.7e+7 & 3.5e+8 & 1e+9 & 1e+10  \\
tech-as-skitter & 6e+17 & 6.4e+15 & 7.1e+14 & 1.3e+15 & 7.3e+12 & 2.7e+14 & 7e+14 & 3e+11 & 3.6e+11 & 3.3e+12 & 2.9e+14  \\
web-wiki-ch-internal & 1.7e+17 & 8e+15 & 3.8e+14 & 1e+15 & 1.5e+13 & 2.2e+14 & 5.2e+13 & 2.4e+11 & 6.5e+11 & 2.6e+12 & 3.9e+13  \\ 
web-hudong & 7.2e+16 & 1.2e+15 & 6.1e+13 & 4.2e+13 & 3.2e+11 & 1.2e+13 & 2.2e+12 & 3.7e+9 & 3.2e+9 & 3.2e+10 & 8e+11  \\
wiki-user-edits-page & 2.5e+18 & 1e+18 & 2.7e+16 & 4.7e+17 & 1.1e+13 & 6.7e+16 & 1.4e+16 & 1.2e+11 & 1.4e+12 & 5.5e+12 & 3.8e+16  \\
web-baidu-baike & 7.6e+17 & 4.5e+15 & 1.5e+14 & 2.8e+14 & 1.4e+12 & 4.4e+13 & 1.9e+13 & 2e+10 & 2.6e+10 & 1e+11 & 5.4e+12  \\
tech-ip & 4.1e+18 & 4.2e+17 & 4.6e+16 & 1.4e+14 & 1.2e+12 & 9.7e+12 & 3e+16 & 1.7e+11 & 8.2e+8 & 1.3e+11 & 9.5e+11  \\
orkut & 4.5e+17	& 8.6e+16	& 1.0e+16& 	7.1e+15	& 6.0e+14	& 2.5e+15	& 6.3e+14& 	1.3e+13	& 1.8e+13	& 8.2e+13	& 3.9e+14 \\
LiveJournal & 1.9e+16	& 2.2e+15& 	3.1e+14	& 1.7e+14	& 2.2e+13& 	5.4e+13	& 1.5e+13	& 3.7e+11& 	1.2e+12& 	3.8e+12	& 1.0e+13	\\ \hline 
%ia-dbpedia-team-bi & 6.3e+11 & 4.2e+10 & 1.3e+10 & 8.7e+7 & 2.6e+6 & 7.5e+7 & 3.6e+8 & 2.2e+5 & 3.9e+3 & 2.9e+4 & 2.3e+6  \\
%web-wikipedia2009 & 2.6e+12 & 7.6e+11 & 2.3e+11 & 4e+10 & 1.8e+10 & 4.9e+10 & 3.9e+10 & 6e+8 & 4.3e+8 & 6.4e+9 & 1.1e+10  \\
%ia-wiki-Talk-dir & 2.7e+17 & 6e+14 & 3.7e+13 & 1.3e+13 & 8.9e+11 & 9e+12 & 9.3e+11 & 1.3e+10 & 8.5e+9 & 1.8e+11 & 3.2e+11  \\
%flickr & 4.6e+16 & 2.2e+15 & 1.3e+14 & 6.7e+14 & 8.9e+12 & 2.1e+14 & 2e+13 & 2e+11 & 7.4e+11 & 4.4e+12 & 5.8e+13  \\
%orkut & 3.6e+17 & 1.2e+16 & 8.8e+14 & 1.4e+15 & 1.5e+13 & 7.1e+14 & 1.2e+14 & 2.6e+11 & 5.5e+11 & 6.9e+12 & 1.2e+14  \\
%italy-osm & 1e+5 & 4e+6 & 1.2e+7 & 4.8e+3 & 3.8e+4 & 3.1e+4 & 2.3e+5 & 3.3e+4 & 3.8e+2 & 5.6e+2 & 4.4e+2  \\
\hline
\end{tabular}
\label{tab:5vertex-counts1}
\end{table*}

\begin{table*}
\caption{\label{tab:5vertex2}5-vertex patterns;  Patterns 13--21 and runtimes in seconds} 
\scriptsize
\begin{tabular}{l|cccccccccc| }
& P12& P13& P14& P15& P16& P17& P18& P19& P20& P21\\\hline 
soc-brightkite & 1.2e+8 & 1.4e+7 & 1.4e+8 & 6.2e+8 & 2.4e+8 & 2.4e+7 & 1.5e+7 & 1.9e+8 & 6.2e+7 & 1.9e+7  \\
tech-RL-caida & 9.8e+7 & 4.9e+9 & 3e+8 & 1.4e+8 & 6.4e+7 & 4.1e+7 & 5.1e+6 & 3.1e+7 & 5.8e+6 & 6.5e+5  \\
Flickr & 5.8e+11	& 3.4e+11	& 1.8e+12	& 1.6e+12	& 6.9e+11	& 7.2e+10& 	1.8e+10& 	1.9e+11	& 1.5e+10	& 7.1e+08 \\
ia-email-EU-dir & 2.7e+8 & 8.3e+7 & 3.8e+8 & 1.1e+9 & 2e+8 & 4.2e+7 & 9.8e+6 & 7e+7 & 9.9e+6 & 1.1e+6 \\
ca-coauthors-dblp & 1.1e+10 & 2e+7 & 2.6e+10 & 3.2e+12 & 5.2e+10 & 7.6e+8 & 3.2e+8 & 3.5e+11 & 3.6e+10 & 5.7e+11 \\
web-google-dir & 9.7e+7 & 3.5e+8 & 8.8e+8 & 3.2e+8 & 8.7e+7 & 2.5e+7 & 9e+6 & 3.7e+7 & 8.7e+6 & 1.4e+6  \\
tech-as-skitter & 4.6e+11 & 1.6e+14 & 7.2e+13 & 1.3e+12 & 5e+11 & 1.3e+11 & 2.4e+10 & 2.9e+11 & 3.8e+10 & 1.2e+9  \\
web-wiki-ch-internal & 3.9e+11 & 2.7e+12 & 1e+12 & 4.3e+11 & 2e+11 & 4e+10 & 3.5e+9 & 3.1e+10 & 1.3e+9 & 4e+7   \\
web-hudong & 3.9e+9 & 1.3e+11 & 1.2e+10 & 7.9e+9 & 5.4e+9 & 2.4e+9 & 2.6e+9 & 6e+9 & 4.6e+9 & 1.2e+9   \\
wiki-user-edits-page & 5.7e+11 & 1.6e+15 & 2.2e+15 & 3.4e+12 & 1.2e+12 & 6.6e+10 & 1.5e+10 & 4.6e+11 & 2e+10 & 1.4e+7 \\
web-baidu-baike & 1.5e+10 & 1.3e+12 & 9.4e+10 & 8e+9 & 4.4e+9 & 4.9e+8 & 2.9e+7 & 3.4e+8 & 6.9e+6 & 9.9e+4  \\
tech-ip & 2.5e+11 & 3.2e+15 & 2.7e+9 & 4.4e+4 & 7.9e+8 & 4.7e+10 & 1.5e+8 & 2.1e+2 & 0 & 0  \\
orkut & 1.5e+13	& 1.1e+13	& 9.7e+12& 	1.2e+13& 	7.1e+12& 	1.4e+12& 	2.4e+11& 	1.3e+12& 	1.4e+11& 	1.6e+10 \\
LiveJournal &5.0e+11& 	3.0e+11& 	7.1e+11& 	2.3e+12& 	7.1e+11& 	1.6e+11& 	1.3e+11& 	1.0e+12& 	6.3e+11& 	4.7e+11 \\ \hline
%

%ia-dbpedia-team-bi & 1.8e+4 & 2.7e+6 & 1.4e+4 & 0 & 2.2e+2 & 2.5e+2 & 1 & 0 & 0 & 0 & 0.88 \\
%web-wikipedia2009 & 1.7e+9 & 2.4e+9 & 7e+8 & 4.2e+8 & 6.2e+8 & 5.8e+8 & 1.3e+8 & 1e+8 & 1.7e+7 & 1.7e+6  & 50.7 \\
%ia-wiki-Talk-dir & 1.5e+10 & 3e+9 & 1.8e+9 & 2.1e+10 & 5.8e+9 & 1.4e+9 & 2.3e+8 & 9.3e+8 & 8.9e+7 & 4.5e+6 & 264.66 \\
%flickr & 5.8e+11 & 3.4e+11 & 1.8e+12 & 1.6e+12 & 6.9e+11 & 7.2e+10 & 1.8e+10 & 1.9e+11 & 1.5e+10 & 7.1e+8  & 917.57\\
%orkut & 8.1e+11 & 2.9e+12 & 3.4e+12 & 1.5e+12 & 6.3e+11 & 1.4e+11 & 3.2e+10 & 1.5e+11 & 1.7e+10 & 1.8e+9  & 3359.56\\
%italy-osm & 1.2e+3 & 21 & 0 & 0 & 28 & 2 & 4 & 0 & 0 & 0  & 5.59 \
%inf-roadNet-PA & 2.9e+4 & 1.6e+2 & 1 & 35 & 6.8e+2 & 70 & 57 & 1 & 0 & 0  & 1.88\\\
%-roadNet-CA & 5.4e+4 & 4.6e+2 & 1 & 1.2e+2 & 1.6e+3 & 1.4e+2 & 3.1e+2 & 1 & 0 & 0  & 3.47 \\
%-road-usa & 1.1e+5 & 4.2e+2 & 1 & 1.7e+2 & 1e+3 & 4.2e+2 & 2.9e+2 & 1 & 0 & 0 & 45.98 \\

\end{tabular}
\label{tab:5vertex-counts2}
\end{table*}

\end{document}